\documentclass[aps,pre,twocolumn,amsfonts,amssymb,floatfix, superscriptaddress]{revtex4}
\usepackage[T1]{fontenc}
\usepackage[utf8]{inputenc}
\usepackage{float}
\usepackage{graphicx}
\usepackage{amssymb}
\usepackage{amsthm}
\usepackage{amsmath}
\usepackage{enumerate}
\usepackage{bbm}

\makeatletter

\newcommand*{\diff}{\mathop{\!\mathrm{d}\!}}
\newcount\hour \newcount\minute
\hour=\time \divide \hour by 60 \minute=\time \count99=\hour
\multiply \count99 by -60 \advance\minute by \count99 \hour=\time
\divide \hour by 60
\date{September 2} 

\newtheorem{lemma}{Lemma}

\makeatother

\begin{document}

\title{Discontinuous transitions in globally coupled potential systems with additive noise}

\author{Rüdiger Kürsten}
\affiliation{\mbox{Institut für Theoretische Physik, Universität Leipzig, POB 100 920, D-04009 Leipzig, Germany and} \\
\mbox{International Max Planck Research School Mathematics in the Sciences, Inselstraße 22, D-04103 Leipzig, Germany} }
\affiliation{\mbox{Institut für Physik, Universität Greifswald, Felix-Hausdorff-Str. 6, 17487 Greifswald, Germany}}

\author{Ulrich Behn}
\affiliation{\mbox{Institut für Theoretische Physik, Universität Leipzig, POB 100 920, D-04009 Leipzig, Germany and} \\
\mbox{International Max Planck Research School Mathematics in the Sciences, Inselstraße 22, D-04103 Leipzig, Germany} }

\begin{abstract}
An infinite array of globally coupled overdamped constituents moving in a double-well potential with $n$-th order saturation term under the influence of additive Gaussian white noise is investigated. The system exhibits a continuous phase transition from a symmetric phase to a symmetry-broken phase. The qualitative behavior is independent on $n$. The critical point is calculated for strong and for weak noise, these limits are also bounds for the critical point. Introducing an additional nonlinearity, such that the potential can have up to three minima, leads to richer behavior. There the parameter space divides in three regions, a region with a symmetric phase, a region with a phase of broken symmetry and a region where both phases coexist. The region of coexistence collapses into one of the others via a discontinuous phase transition whereas the transition between the symmetric phase and the phase of broken symmetry is continuous. The tricritical point where the three regions intersect, can be calculated for strong and for weak noise. These limiting values form optimal bounds on the tricritical point. 
In the region of coexistence simulations of finite systems are performed.
One finds that the stationary distribution of finite but large systems differs qualitatively from the one of the infinite system. Hence the limits of stationarity and large system size do not commute.
\end{abstract}
\maketitle

\section{Introduction}

Noise in nonlinear systems can induce many nontrivial effects such as for example stochastic resonance \cite{GHJM98}, stochastic transport \cite{Reimann02} or noise induced transitions \cite{HL84}. 
In spatially extended systems \cite{SSG07, GS99} the dynamics of one system site depends on its neighboring sites.
In this paper we deal with global coupling, where each site is coupled to each other which might be interpreted as a mean field approximation of systems with local coupling.
We consider anharmonic systems with additive Gaussian white noise 
\begin{align}
	\dot{x_i} = & -\partial_{x_i}U(\mathbf{x}) + \xi_i(t),
\label{eq:langevingeneral}
\end{align}
where $i=1, \dots, L$ and $\mathbf{x}$ denotes the vector of all $x_i$. The potential $U$ consists of a single-particle potential $U_0$, that is felt by each particle, and a two-particle interaction potential $U_{\text{int}}$, that is felt by each pair of particles and that depends only on the distance between the particles
\begin{align}
	U(\mathbf{x})= \sum_{i=1}^{L} U_{0}(x_i) + \frac{1}{L}\sum_{1\le i < j \le L} U_{\text{int}}(x_i-x_j).
	\label{eq:potentialgeneral}
\end{align}
The $\xi_i$ are delta correlated Gaussian white noise of strength $\sigma$
\begin{align}
	\langle \xi_i(t) \xi_j(s)\rangle = \sigma^{2}\delta_{ij} \delta(t-s),
	\label{eq:noise}
\end{align}
where $\langle \cdot \rangle$ denotes the expectation value.

We are interested in nonlinear systems, that means either $U_0$ or $U_{\text{int}}$ are anharmonic potentials. In this paper we examine only harmonic coupling
\begin{align}
	U_{\text{int}}(r)= \frac{D}{2}r^{2},
	\label{eq:linearcoupling}
\end{align}
with coupling constant $D$ and anharmonic single-particle potentials $U_{0}$.

The system \eqref{eq:langevingeneral} can equivalently be described by the Fokker-Planck equation
\begin{align}
	\partial_{t} p(\mathbf{x}, t)= \sum_{i=1}^{L} - \partial_{x_{i}} \Big\{ \Big[ \sum_{j=1}^{L} -\partial_{x_{j}} U(\mathbf{x}) -\partial_{x_{j}}\frac{\sigma^{2}}{2} \Big]p(\mathbf{x}, t) \Big\},
	\label{eq:fokkerplanckgeneral}
\end{align}
which has the stationary solution
\begin{align}
	p_{\text{s}}(\mathbf{x})&= \frac{1}{Z} \exp\left[ -  \frac{2}{\sigma^{2}} U(\mathbf{x}) \right]
	\label{eq:statsolutiongeneral}
\end{align}
with normalization
\begin{align}
	Z&= \int_{\mathbb{R}^{L}}\diff \mathbf{x} \exp\left[ - \frac{2}{\sigma^{2}} U(\mathbf{x}) \right].
\end{align}

The empirical measure defined by
\begin{align}
	p_{\text{e}}^{L}(A,t):= \frac{1}{L}\sum_{i=1}^{L} \mathbbm{1}_{A}(x_{i}(t))
	\label{eq:empiricalmeasure}
\end{align}
indicates how many particles are in the set $A$ at time $t$.
It converges for $L \rightarrow \infty$ to a probability measure defined by the density $p(x, t)$, which is the solution of the nonlinear Fokker-Planck equation \cite{Frank05, Shiino85}
\begin{align}
	\partial_{t}p(x, t) =&  -\partial_{x}\Big\{ \Big[ -\partial_{x} U_{0}(x) + D\int_{-\infty}^{\infty}x' p(x',t) \diff x' \notag \\
	&-Dx  -\frac{\sigma^{2}}{2} \partial_{x} \Big]p(x,t) \Big\}.
	\label{eq:nonlinearfokkerplanck}
\end{align}
The stationary solution is
\begin{align}
	p_{\text{s}}(x, m)= \frac{1}{Z}\exp\left[ -\frac{2}{\sigma^{2} } \big(U_{0}(x) + \frac{1}{2} D x^{2} -Dmx\big)  \right]
	\label{eq:statsolutionnonlinearfp}
\end{align}
with normalization
\begin{align}
	Z=\int_{-\infty}^{\infty}\diff x \exp\left[ -\frac{2}{\sigma^{2} } \big(U_{0}(x) + \frac{1}{2} D x^{2} -Dmx\big)  \right],
	\label{eq:statnormalization}
\end{align}
where the mean field $m$ has to satisfy the self-consistency equation \cite{DZ78}
\begin{align}
	m= \int_{-\infty}^{\infty} \diff x x p_{\text{s}}(x, m).
	\label{eq:selfconsistency}
\end{align}
One can reformulate this condition on $m$ using the self-consistency map
\begin{align}
	F(m):= \int_{-\infty}^{\infty}\diff x' x'p_{\text{s}}(x', m).
	\label{eq:selfconsistencymap}
\end{align}
Fixed points of this map satisfy the self-consistency condition \eqref{eq:selfconsistency}.
For certain types of single-particle potentials $U_0$ we analyze these fixed points in detail.

In Sec.~\ref{sec:two} we consider double-well potentials.
Systems with such potentials but with multiplicative noise have been studied, e.g., in \cite{GHS93, BLMKB02} and with both, additive and multiplicative noise, e.g., in \cite{BPAH94, GPSB96}.
The model with pure additive noise and the single-particle potential $U_0= (a/2)x^2-(1/4)x^4$ was introduced by Kometani and Shimizu \cite{KS75} and later on it was intensively studied \cite{DZ78, Dawson83, Shiino85, Shiino87, BPAH94, GPSB96, KGB13}.
It exhibits a continuous phase transition from a symmetric phase to a phase of broken symmetry at a unique critical point \cite{Dawson83, KGB13}.

We show that for double-well potentials with higher order saturation term the qualitative behavior of the infinite system subject to additive noise is the same.
There exists a continuous symmetry breaking phase transition at a unique critical point.
We prove existence and uniqueness of the critical point and calculate optimal lower and upper bounds for the critical point that are assumed in the limits of strong and weak noise.

In Sec.~\ref{sec:three}, we deal with more complicated single-particle potentials that can have up to three local minima.
Potentials of this type in a less general form have been studied in \cite{BPT94, LZS98, ZGS99} but with multiplicative and additive noise.
We find that for the infinite system the phase space divides into three regions, a phase with only stable symmetric solutions, a phase of broken symmetry and a phase where stable symmetric and symmetry-broken solutions coexist.
The transition between symmetric and symmetry-broken phase is continuous.
On the other hand the transitions between the region of coexistence and one of the other phases are discontinuous.
A noise induced discontinuous transition was first reported \cite{MLKB97}.
All three phases meet in the tricritical point.
We prove optimal bounds for the tricritical point that are assumed asymptotically for strong and for weak noise.

We further sample the stationary center of mass distribution for systems of finite size in the region of coexistence using patchwork sampling \cite{KB16}.
We find that the two limits of infinite observation time (or stationarity) and infinite system size do not commute.
If the limit of infinite system size is performed first there are three stable stationary solutions from which two are breaking the symmetry and one is symmetric.
If on the other hand the limit of infinite observation time is performed first the distribution is always symmetric. Depending on parameters it consists either of a peak at zero or of two peaks located at the nonzero fixed points of the mean field of the infinite system.

Sec.~\ref{sec:conclusions} summarizes our main conclusions and gives a brief outlook.
Technical discussions and mathematical proofs have been moved to the Appendices A-D.
\section{Double-Well Potential\label{sec:two}}
The case
\begin{align}
	U_{0}(x)= -\frac{a}{2} x^{2} + \frac{1}{4} x^{4} 
	\label{eq:nonlinearsystem1}
\end{align}
with harmonic interaction \eqref{eq:linearcoupling}
was intensively studied, e.g., in \cite{KS75, DZ78, Dawson83, Shiino85, Shiino87, KGB13}. The one-particle potential is a double-well potential where the barrier is tuned by the parameter $a$. The interaction is harmonic with coupling strength $D$.
The system shows a phase transition in the thermodynamic limit, $L\rightarrow \infty$.
Given two of the parameters $a, D, \sigma>0$, there exists a unique critical value of the third one that separates a symmetric phase from a phase of broken symmetry \cite{KGB13}.
In the symmetric phase, as Shiino \cite{Shiino87} showed, $p_{\text{s}}(x, m=0)$ is the only stationary solution of Eq.~\eqref{eq:nonlinearfokkerplanck}. In the symmetry-broken phase $p_{\text{s}}(x, m=0)$ is an unstable solution of Eq.~\eqref{eq:nonlinearfokkerplanck} and $p_{\text{s}}(x, m=m_{+})$ and $p_{\text{s}}(x, m=-m_{+})$ are the two stable solutions of Eq.~\eqref{eq:nonlinearfokkerplanck}, where $\pm m_{+}$ are the two nonzero fixed points of the self-consistency map \eqref{eq:selfconsistencymap}, cf. \cite{Shiino87}.

We are going to generalize the results of \cite{KGB13} for higher order saturation terms considering
\begin{align}
	U_{0}(x)&= -\frac{a}{2} x^{2} + \frac{1}{n} x^{n}
	\label{eq:nonlinearsystem2}
\end{align}
for $n=6, 8, 10, \dots$ and the same harmonic interaction potential \eqref{eq:linearcoupling}.
It is possible to reduce the number of parameters by one through a rescaling of space and time. However we keep all parameters since they have a physical meaning as noise strength, coupling strength and potential barrier height.

\subsection{Self-consistency map}

The self-consistency map \eqref{eq:selfconsistencymap} has for the above choice of $U_{0}$ the following properties, cf. Fig.~\ref{fig:selfcons}
\begin{align}
	F^{(2k)}(0)&=0,
	\label{eq:selfconsproperties1}\\
	F(m)&=-F(-m),
	\label{eq:selfconsproperties2}\\
	F''(m)&\lesseqgtr0 \text{ for } m\gtreqless0,
	\label{eq:selfconsproperties3}\\
	F'(m)&\rightarrow 0 \text{ for }m\rightarrow\infty.
	\label{eq:selfconsproperties4}
\end{align}
Equations~(\ref{eq:selfconsproperties1},\ref{eq:selfconsproperties2}) follow directly from the definition \eqref{eq:selfconsistencymap}, the inequalities~\eqref{eq:selfconsproperties3} are a consequence of ferromagnetic inequalities
\footnote{We will state a special case of the Griffith-Hurst-Sherman inequality. Proofs for spin systems can be found in \cite{GHS70, Lebowitz74, EM75}, for continuous systems such as used here, proofs are given in \cite{EMN76, EN78}, cf. also \cite{Dawson83}. Let $p(x)=\exp[c_1 -\lambda x^{n}]$ with arbitrary constant $c_1$, $\lambda >0$ and $n=4, 6, 8, 10, \dots$, then $\diff\,^{3}/\diff h^{3} \left[ \ln \int_{-\infty}^{\infty}\diff x \exp(hx)\exp(-\gamma x^{2})p(x) \right]\le 0$, for arbitrary $\gamma$ and $h\ge 0$. $\lambda= 2/(\sigma^{2}n)$, $\gamma= (D-a)/\sigma^{2}$ and $h=Dm$ yields the desired result \eqref{eq:selfconsproperties3} for $m>0$, the case $m<0$ follows by symmetry.} 
and \eqref{eq:selfconsproperties4} can be directly obtained by asymptotic evaluation of the integral occurring in \eqref{eq:selfconsistencymap}.
Furthermore the function $F(\cdot)$ is continuous and arbitrarily often continuously differentiable.
For positive $a, D, \sigma$ $F$ is continuous and arbitrarily often continuously differentiable also with respect to these parameters.

\begin{figure}
	\includegraphics{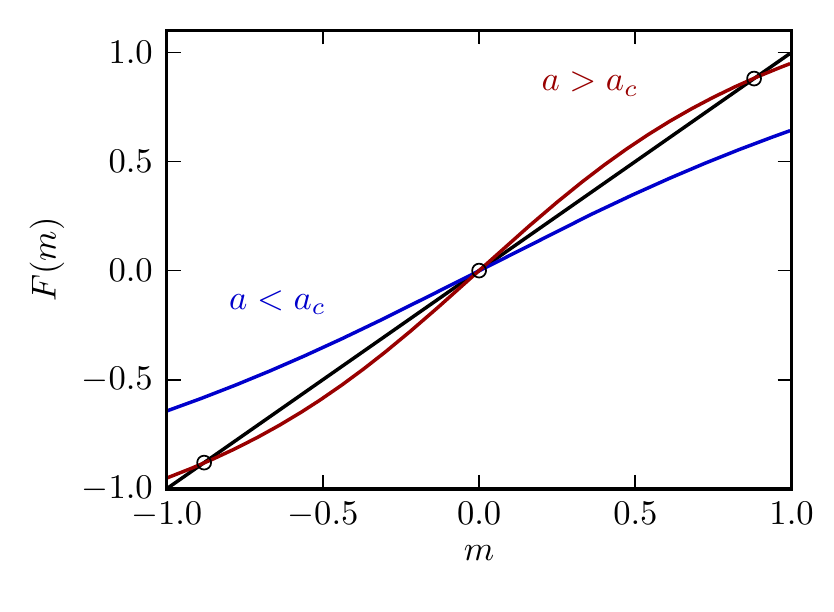}
	\caption{(color online) Self-consistency map \eqref{eq:selfconsistencymap} for $n=4$, $\sigma=1$, $D=1$ and $a=0.5<a_c$ (blue line) and for $a=1.5>a_c$ (red line). Fixed points of the map occur at the intersection with the straight line and are indicated by circles. In the symmetric phase (blue line) there is only one fixed point at $m=0$ which is stable. In the symmetry-broken phase (red line) $m=0$ is unstable and there exist two unstable fixed points $m_{+}$ and $m_{-}=-m_{+}$. One of these fixed points is selected by initial conditions.\label{fig:selfcons}}
\end{figure}

Equation \eqref{eq:selfconsproperties1} states that $m=0$ is always a fixed point of the self-consistency map. If $F'(0)\le 1$ we conclude from \eqref{eq:selfconsproperties3} that $m=0$ is the only fixed point, it is stable.
If on the other hand $F'(0)>1$ $m=0$ is an unstable fixed point and we follow from  (\ref{eq:selfconsproperties3}, \ref{eq:selfconsproperties4}), that there is exactly one positive fixed point $m_{+}$ of Eq.~\eqref{eq:selfconsistencymap} and due to the symmetry \eqref{eq:selfconsproperties2} there is also exactly one negative fixed point $m_{-}$ with $m_{+}=-m_{-}$. These nonzero fixed points are stable and break the reflection symmetry $x \leftrightarrow -x$.
This argumentation follows Dawson \cite{Dawson83} who stated equivalent results for the case $n=4$.

\subsection{Phase transition condition}

Stable fixed points of the self-consistency map \eqref{eq:selfconsistencymap} correspond to stable stationary solutions of the nonlinear Fokker-Planck equation \eqref{eq:nonlinearfokkerplanck} \cite{Shiino87}. Hence the transition from symmetric to symmetry-broken solutions occurs when 
\begin{align}
	F'(0)=1 \qquad \Leftrightarrow \qquad \langle x^2\rangle_{m=0} = \frac{\sigma^2}{2D},
	\label{eq:PTCmomentrelation}
\end{align}
where $\langle \cdot \rangle_{m=0}$ denotes the expectation value with respect to $p_{\text{s}}(x, m=0)$, cf. Eq.~\eqref{eq:statsolutionnonlinearfp}.

\subsection{The critical manifold}

We consider for the moment $D, \sigma$ to be fixed parameters and $a$ as a control parameter.
For $D\le 0$ there is no symmetry breaking as can be seen from Eq.~\eqref{eq:PTCmomentrelation}, therefore we consider only $D>0$.
For the symmetric solution, the second moment as a function of $a$ has the following properties
\begin{flalign}
	\lim_{a\rightarrow - \infty} \langle x^{2}\rangle_{m=0}(a) &= 0,
	\label{eq:secondmomentproperties1}\\
	\lim_{a\rightarrow  \infty} \langle x^{2}\rangle_{m=0}(a) &= +\infty,
	\label{eq:secondmomentproperties2}\\
	\partial_{a} \langle x^{2}\rangle_{m=0}(a) &= \frac{1}{\sigma^{2}}\left( \langle x^4\rangle_{m=0}-\langle x^{2}\rangle_{m=0}^{2} \right)
	\notag
	\\
	= \frac{1}{\sigma^{2}}\big[ \langle (x^{2}- & \langle x^{2}\rangle_{m=0})^{2}  \rangle_{m=0} \big]>0.
	\label{eq:secondmomentproperties3}
\end{flalign}
The first two properties follow from an asymptotic evaluation of the expectation value. The third property follows from direct calculation.
From the properties (\ref{eq:secondmomentproperties1}, \ref{eq:secondmomentproperties2}) it follows by mean value theorem that there exists a critical value $a_c$ such that
\begin{align}
	\langle x^{2}\rangle_{m=0}\big(a_{c}(\sigma,D), \sigma, D \big)=\frac{\sigma^{2}}{2D}.
	\label{eq:acritical}
\end{align}
From the monotonicity \eqref{eq:secondmomentproperties3} we conclude that this critical point is unique. Furthermore we find that
\begin{align}
	F'(0) \lesseqgtr 1 \text{ for } a \lesseqgtr a_{c}.
	\label{eq:fprimeofzero}
\end{align}
Hence below and at the critical point $a\le a_{c}$ the system is in the symmetric phase and above the critical point $a>a_{c}$ the system is in one of the symmetry-broken phases.

We consider $a_{c}=a_{c}(\sigma^{2}, D)$ as a function of noise strength and coupling strength that is well defined for $\sigma,D>0$.
We will show that $a_{c}(\sigma^{2}, D=D_0)$ considered as a function of $\sigma^{2}$ for arbitrary but fixed $D_0>0$, assumes all positive values. Equivalently for some fixed $\sigma_{0}>0$, $a_{c}(\sigma_{0}^{2}, D)$ considered as a function of $D$, assumes all positive values as well.
It follows that for fixed positive $a, \sigma_0$ there exists a critical value $D_{c}$ and for fixed positive $a, D_0$ there exists a critical value $\sigma_{c}$, such that
\begin{align}
	\langle x^{2}\rangle_{m=0}\big(a, \sigma, D_c(a, \sigma)\big)&= \frac{\sigma^{2}}{2 D_{c}(a, \sigma)},
	\label{eq:secondmomentcritD}\\
	\langle x^{2}\rangle_{m=0}\big(a, D_c, \sigma_{c}(a, D)\big)&= \frac{\sigma_{c}^{2}(a, D)}{2 D}.
	\label{eq:secondmomentcritsigma}
\end{align}
We will further show that $a_c(\sigma^{2}, D)$ is invertible with respect to each argument from which follows that the critical values $\sigma_{c}(a, D)$, $D_{c}(a, \sigma)$ are unique.

Since $a_{c}(\sigma^{2}, D)$ is the zero of a, with respect to both arguments, continuous and arbitrarily often continuously differentiable function, cf. Eq.~\eqref{eq:acritical}, it is continuous and arbitrarily often continuously differentiable with respect to both arguments itself.
Hence to show the invertibility it suffices to show that $a_{c}(\sigma^{2}, D)$ is monotone with respect to both arguments.
The monotonicity is proved in Appendix \ref{app:A}.

In the next subsection we will show that
\begin{align}
	\left( \frac{\sigma^{2}}{2D}\right)^{n/2-1}< a_{c}(\sigma,D)<(n-1)!!\left( \frac{\sigma^{2}}{2D}\right)^{n/2-1},
	\label{eq:boundsonac}
\end{align}
where $(n-1)!!$ denotes the product of all odd positive integers from one to $n-1$.
We conclude from these bounds that if one of the parameters $\sigma, D>0$ is fixed, the infimum of $a_c$ as a function of the other parameter is zero and the supremum is infinite. Since $a_c$ is continuous it follows that it assumes all positive values.

Hence we have shown that the critical manifold is well behaved, that is given two of the system parameters $a, D, \sigma$, positive, there exists a unique critical value for the third parameter that separates the symmetric phase from the phase of broken symmetry. The qualitative behavior of the system is the same as for $n=4$ \cite{KGB13}.

We can develop the self-consistency map around the critical point to obtain the normal form
\begin{align}
	F(m) \approx F'(0) m + \frac{1}{6} F'''(0) m^{3}.
	\label{eq:normalform1}
\end{align}
Note that $F'''(0)= (2D/\sigma^2)^3 \kappa_{4}<0$ at the critical point, where the fourth cumulant $\kappa_4$ is negative as shown in Appendix \ref{app:B}, cf. Eq.~\eqref{eq:kurtosis}. Close to the critical point we have $F'(0)= 1 + c (a-a_c) + \mathcal{O}( (a-a_c)^2)$, with some constant $c$. Hence the stable fixed points grow with exponent $1/2$ close to the critical point. 

\subsection{Optimal bounds for the critical point}

From the symmetric stationary solution
\begin{align}
	p_{\text{s}}(x, m=0)= \frac{1}{Z} \exp \left[ \frac{2}{\sigma^{2}} \left(\frac{a-D}{2}x^{2} -\frac{1}{n} x^{n} \right) \right],
	\label{eq:ps}
\end{align}
cf. Eq.~\eqref{eq:statsolutionnonlinearfp}, we find for $k=0, 1, 2, \dots$
\begin{align}
	&Z_{m=0} \langle x^k \rangle_{m=0} =  \int_{-\infty}^{\infty} \diff x  x^k \exp \left[ \frac{2}{\sigma^{2}} \left(\frac{a-D}{2}x^{2} -\frac{1}{n} x^{n} \right) \right]
	\notag \\ 
	&=  -\frac{1}{k+1} \int_{-\infty}^{\infty} \diff x \frac{2}{\sigma^2} \left( (a-D) x^{k+2} -x^{k+n} \right)
	\notag \\ 
	&\phantom{=} \times \exp \left[- \frac{2}{\sigma^{2}} \left(\frac{a-D}{2}x^{2} -\frac{1}{n} x^{n} \right) \right] 
	\label{eq:relmoments}
\end{align}
and hence
\begin{align}
	\frac{\sigma^{2}}{2}(k+1)\langle x^k\rangle_{m=0} =  \langle x^{k+n}\rangle_{m=0}  -(a-D) \langle x^{k+2}\rangle_{m=0}.
	\label{eq:relmomens2}
\end{align}

As we prove in Appendix \ref{app:B}, for all parameter values it holds
\begin{align}
	\langle x^{n}\rangle_{m=0}-\langle x^{2}\rangle_{m=0}^{n/2}>0,
	\label{eq:momentrelation1}\\
	\langle x^{n}\rangle_{m=0}-(n-1)!!\langle x^{2}\rangle_{m=0}^{n/2}<0.
	\label{eq:momentrelation2}
\end{align}
Inserting the moment relation \eqref{eq:relmomens2} for $k=0$ we find the bounds \eqref{eq:boundsonac} for $a_c$.
These bounds are optimal in the sense that they are asymptotically assumed in the limits of strong and weak noise, respectively, cf. Appendix \ref{app:C}. In Fig.~\ref{fig:ac} the critical control parameter $a_c$ is displayed as a function of $(\sigma^{2}/2D)^{n/2-1}$ for $n=6$. The values of $a_c$ have been obtained by numerical evaluation of Eq.~\eqref{eq:PTCmomentrelation}. All values lie between the bounds \eqref{eq:boundsonac}. For weak noise $a_c$ is close to the upper bound, for strong noise it is close to the lower bound.
\begin{figure}[H]
	\includegraphics{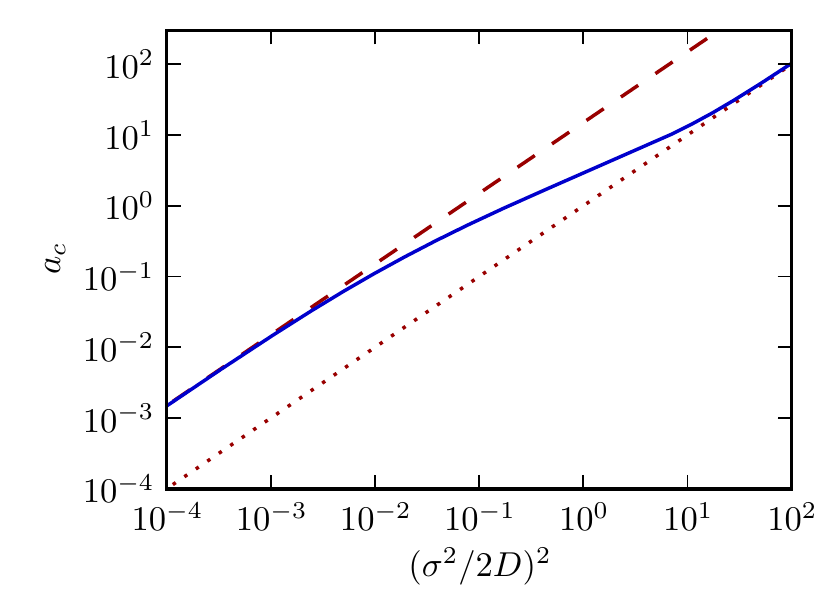}
	\caption{(color online) Critical point $a_c$ as a function of $(\sigma^{2}/2D)^{2}$ for $n=6$ on a double logarithmic scale.
	Numerical results for $D=1$ and different noise strengths are displayed by the solid line (blue).
	Optimal bounds on $a_c$ are given by \eqref{eq:boundsonac}. The upper bound (red dashed line) is asymptotically reached for weak noise, the lower bound (red dotted line) is asymptotically reached for strong noise.\label{fig:ac}}
\end{figure}
\section{Competing Nonlinearities\label{sec:three}}

In the previous section we demonstrated that the qualitative behavior of the system \eqref{eq:nonlinearsystem2} does not depend on the order of the saturation term. However, two competing nonlinear terms may change the principal behavior.
We consider the system
\begin{align}
	U_{0}(x)&= -\frac{a}{2} x^{2}  - \frac{b}{4} x^{4} + \frac{1}{6} x^{6}
	\label{eq:nonlinearsystem3}
\end{align}
with the same interaction potential \eqref{eq:linearcoupling}.

\subsection{Equal sign of nonlinearities}

If $b<0$ the quartic and the sixtic term in Eq.~\eqref{eq:nonlinearsystem3} have the same sign. In this case there is not really a competition between the two terms. Many properties of the system are similar to the system of the previous section, for example the properties (\ref{eq:selfconsproperties1}-\ref{eq:selfconsproperties4}) of the self-consistency map. We conclude from these relations that there is a symmetric phase and a phase of broken symmetry and that there exists a critical point $a_c$ for positive parameters $D,\sigma>0$ that is given by $F'(0)=1$.
The normal form of the system close to the critical point remains
\begin{align}
	F(m) \approx F'(0) m + \frac{1}{6} F'''(0) m^{3}
	\label{eq:normalform2}
\end{align}
and developing $F'(0)= 1 + c (a-a_c) + \mathcal{O}( (a-a_c)^2)$ with some constant $c$ yields the mean field exponent $1/2$ for the stable fixed points close to the critical point.
Essentially in this case the system behaves qualitatively as the systems considered above.

\subsection{Opposite sign of nonlinearities}

If $b>0$ there is a competition between the negative quartic term and the positive sixtic term in Eq.~\eqref{eq:nonlinearsystem3}, which leads to new effects.
The properties (\ref{eq:selfconsproperties1}, \ref{eq:selfconsproperties2}, \ref{eq:selfconsproperties4}) of the self-consistency map remain valid also for \eqref{eq:nonlinearsystem3}, but the property \eqref{eq:selfconsproperties3} is not true in general.
Therefore the self-consistency map can exhibit more complex behavior and it is not straightforward to generalize the analysis of the previous sections.
From a numerical evaluation of the self-consistency map we have observed three characteristic cases that are presented in Fig.~\ref{fig:phasediagram}b. In phase I the stable fixed point $m=0$ is the only fixed point of the self-consistency map. In phase II $m=0$ is an unstable fixed point and there are two nonzero stable fixed points $m_{+}=-m_{-}$. There is another phase, phase III, where stable symmetric and symmetry-broken solutions coexists. Hence in phase III $m=0$ is a stable fixed point, there are two unstable fixed points $m_{+,\text{us}}=-m_{-,\text{us}}$ and two more stable fixed points $m_{+,\text{s}}=-m_{-,\text{s}}$.

From phase I, changing parameters one can reach phase II via a supercritical pitchfork bifurcation and phase III via a saddle-node bifurcation, when a new fixed point $m^{*}$ with $F(m^{*})=m^{*}$ and $F'(m^{*})=1$ appears at the transition. Furthermore phase III is reached from phase II via a subcritical pitchfork bifurcation, where the unstable fixed point $m=0$ gains stability. The pitchfork bifurcations occur at $m=0$ and the transition condition is $F'(0)=1$, cf. Eq.~\eqref{eq:PTCmomentrelation}. For the saddle-node bifurcation it is not so easy to localize $m^{*}$ since its occurrence depends on global properties of the self-consistency map.

Although in general a rigorous analysis of the self-consistency map is difficult we can mathematically found our observation locally. Consider a parameter regime of phase III, where all fixed points are close to $m=0$. In phase III we have $F'(0)<1$ and in any case $F''(0)=0$, cf. \eqref{eq:selfconsproperties3}. In order to have a positive fixed point close to zero we require that $F'''(0)>0$ and to have a second positive fixed point close to it, we need $F'''(\varepsilon)<0$ for some small, positive $\varepsilon$. Therefore, as $F^{(4)}(0)=0$, cf. Eq.~\eqref{eq:selfconsproperties1}, $F^{(5)}(0)<0$ is necessary.
Hence to be in phase III and have all fixed points close to $m=0$ the requirements are
\begin{align}
	F'(0)&<1,
	\\
	F'''(0)&>0,
	\\
	F^{(5)}(0)&<0.
	\label{eq:requirementsphase3}
\end{align}
If the first inequality is reversed we are in phase II, if the second condition fails we are in phase I. Hence a point that satisfies
\begin{align}
	F'(0)&=1,
	\\
	F'''(0)&=0.
	\label{eq:conditiontricritical}
\end{align}
is critical twofold. It is called the tricritical point due to the fact that it lies at the intersection of all three phases. We will show that at the tricritical point and also in a neighborhood of it, the condition \eqref{eq:requirementsphase3} is satisfied. Thus the picture described above is rigorous at least in a neighborhood of the tricritical point.
In order to prove \eqref{eq:requirementsphase3}, in fact we need only the condition \eqref{eq:conditiontricritical}.
Note that the derivatives of the self-consistency map coincide up to a factor with the cumulants of the stationary probability distribution $p_{\text{s}}(x, m)$
\begin{align}
	F^{(k)}(m) = \big( 2D/\sigma^{2} \big)^{k} \kappa_{k+1},
	\label{eq:selfconsistencyderivatives}
\end{align}
where $k=0,1,2,\dots$ and $\kappa_{k}$ is the $k$th cumulant. 
Hence we need to show that 
\begin{align}
\kappa_{6}<0
	\label{eq:negativesixthcumulant}
\end{align}
in a neighborhood of the parameter set where $\kappa_{4}=0$ at $m=0$. The proof is given in Appendix \ref{app:D}. 

Close to the tricritical point we can develop the self-consistency map for small $m$ as
\begin{align}
	F(m)= \alpha m + \beta m^{3} + \gamma m^{5},
	\label{eq:developmentselfcons}
\end{align}
where $\alpha, \beta, \gamma$ are functions of the parameters $a,b,\sigma,D$ and $\gamma<0$ and at the tricritical point $\beta=0$. Fixed points of the self-consistency map are real roots of
\begin{align}
	(\alpha-1) +\beta m^{2} + \gamma m^{4}
	\label{eq:fixedpointsdev}
\end{align}
and $m=0$. For $\alpha\le 1$ and $\beta\le 0$, $m=0$ is the only fixed point and the system is in phase I.
For $\alpha > 1$, $m=0$ is unstable and there are two more fixed points that are stable
\begin{align}
	m_{+/-}=\pm \sqrt{\frac{-\beta}{2\gamma} + \sqrt{ \frac{\beta^{2}}{4\gamma^{2}}+ \frac{1-\alpha}{\gamma}}}.
	\label{eq:phase2}
\end{align}
The system is in phase II.
If $\alpha< 1$ and $\beta \ge 0$ the system is either in phase I or in phase III. If $\frac{\beta^{2}}{4\gamma^{2}} < \frac{\alpha-1}{\gamma}$ the system is in phase I. If $\frac{\beta^{2}}{4\gamma^{2}} > \frac{\alpha-1}{\gamma}$ the system is in phase III. The unstable fixed points are
\begin{align}
	m_{+/-, \text{us}}=\pm \sqrt{\frac{-\beta}{2\gamma} - \sqrt{ \frac{\beta^{2}}{4\gamma^{2}}+ \frac{1-\alpha}{\gamma}}}
	\label{eq:phase31}
\end{align}
and the stable fixed points are $m=0$ and
\begin{align}
	m_{+/-, \text{s}}=\pm \sqrt{\frac{-\beta}{2\gamma} + \sqrt{ \frac{\beta^{2}}{4\gamma^{2}}+ \frac{1-\alpha}{\gamma}}}.
	\label{eq:phase32}
\end{align}
The saddle-node bifurcation occurs when $\alpha<1$, $\beta>0$ and $\frac{\beta^{2}}{4\gamma^{2}} = \frac{\alpha-1}{\gamma}$.
The coefficients in Eq.~\eqref{eq:developmentselfcons} can be expressed in terms of the cumulants according to Eq.~\eqref{eq:selfconsistencyderivatives}. We have
\begin{align}
	\alpha &= \frac{2D}{\sigma^{2}} \langle x^{2}\rangle,
	\notag
	\\
	\beta &= \frac{1}{6}\Big(\frac{2D}{\sigma^{2}} \Big)^{3} \kappa_{4},
	\notag
	\\
	\gamma &= \frac{1}{120}\Big( \frac{2D}{\sigma^{2}}\Big)^{5} \kappa_{6}.
	\label{eq:coeffcumulants}
\end{align}
Hence the condition for the saddle-node bifurcation close to the tricritical point can be rewritten as
\begin{align}
	\langle x^{2}\rangle - \frac{\sigma^{2}}{2D}= \frac{5}{6} \frac{\kappa_{4}^{2}}{\kappa_{6}}.
	\label{eq:saddlenode}
\end{align}
In analogy to the previous cases we can investigate the behavior of the stable fixed points close to the critical point when the parameter $a$ is varied. The leading behavior is 
\begin{align}
	\alpha &= 1 + c_{1} (a-a_{\text{tc}}),
	\\
	\beta &= c_{2}(a-a_{\text{tc}}),
	\\
	\gamma &= \gamma_{0} + c_{3} (a-a_{\text{tc}}),
	\label{eq:critexponent}
\end{align}
where $c_i$ are some constants. In particular $c_1=(2D/\sigma^4)\langle x^{4}\rangle - \langle x^{2}\rangle^{2}>0$. Hence we are in phase II and the leading behavior for the stable fixed points is according to Eq.~\eqref{eq:phase2}
\begin{align}
	m_{\pm} \approx \pm\sqrt{\tilde{c}_{1} (a-a_{\text{tc}}) + \sqrt{ \tilde{c}_2 (a-a_{\text{tc}}) } } \approx \hat{c} (a-a_{\text{tc}})^{1/4}
	\label{eq:critexponent2}
\end{align}
with some positive constants $\tilde{c}_1, \tilde{c}_2, \hat{c}$. Thus the critical exponent is $1/4$ at the tricritical point.

In Fig.~\ref{fig:phasediagram}a we present the phase diagram in the $(a,b)$ plane. The phase boundaries have been obtained by numerical evaluation of the corresponding conditions. The discontinuous transition between phase I and phase III was obtained from numerical evaluations of the full self-consistency map \eqref{eq:selfconsistencymap} (solid blue line) and as a solution of the approximation \eqref{eq:saddlenode} close to the tricritical point which can not be distinguished on the scale of the plot. In Fig.~\ref{fig:phasediagram}b  the self-consistency map is shown with one representative for each phase. 

\subsection{The tricritical point}

The point where all three phases meet is called the tricritical point. It is determined as the intersection point of the lines defined by
\begin{align}
	F'(0)=1, \qquad F'''(0)=0
	\label{eq:definitiontricriticalpoint}
\end{align}
which is equivalent to
\begin{align}
	\langle x^{2}\rangle= \frac{\sigma^{2}}{2D}, \qquad \langle x^{4}\rangle-3\langle x^{2}\rangle^{2}=0,
	\label{eq:tricriticalpointcondition}
\end{align}
confer Eq.~(\ref{eq:selfconsistencyderivatives}). We used that $m=0$ is the only fixed point of the self-consistency map at the tricritical point and hence all odd moments are zero.
Therefore at the tricritical point we know the second and the fourth moment. Analogously to Eq.~\eqref{eq:relmomens2} we find for $k=0,1,2, \dots$ 
\begin{align}
	\frac{\sigma^{2}}{2}(k+1)\langle x^{k}\rangle_{m=0}=&\langle x^{k+6}\rangle_{m=0} - b\langle x^{k+4}\rangle_{m=0}
	\label{eq:momentrelationtwononlinearities}\\
	&-(a-D)\langle x^{k+2}\rangle_{m=0} \notag
\end{align}
from partial integration in the definition of the $k$th moment.
Hence at the tricritical point all moments can be calculated from Eqs.~(\ref{eq:tricriticalpointcondition}, \ref{eq:momentrelationtwononlinearities}). In particular we find
\begin{align}
	\langle x^{6}\rangle_{\text{tc}}=3b\Big(\frac{\sigma^{2}}{2D}\Big)^{2} + a \frac{\sigma^{2}}{2D}.
	\label{eq:sixthmoment}
\end{align}
For any extended probability distribution we have due to the Cauchy-Schwarz inequality
\begin{align}
	\langle x^{4}\rangle^{2} < \langle x^{2}\rangle\langle x^{6}\rangle,
	\label{eq:cauchyschwarz}
\end{align}
which leads with Eq.~\eqref{eq:tricriticalpointcondition} to
\begin{align}
	\langle x^{6}\rangle_{\text{tc}} - 9\langle x^{2}\rangle^{3}_{\text{tc}} > 0 .
	\label{eq:inequalitysixthmoment1}
\end{align}
\begin{figure}[H]
	\includegraphics{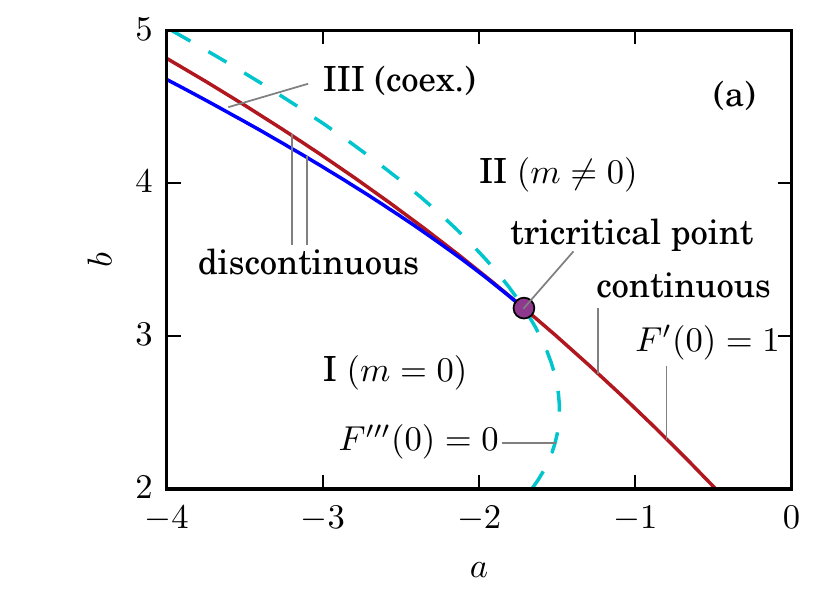}
	\includegraphics{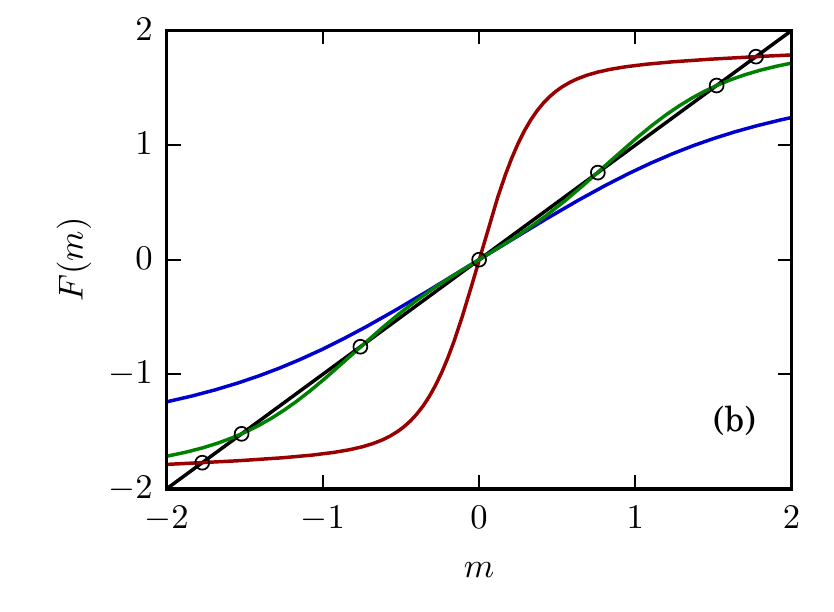}
	\caption{(color online) Phase diagram and self-consistency map for competing quartic and sixtic nonlinearities, cf. Eq.~\eqref{eq:nonlinearsystem3}, for $D=\sigma=1$.
	(a) Phase diagram of the infinite system in the $(a,b)$-plane. The lines indicating a continuous transition (red, between phase I and phase II) and a discontinuous transition (blue, red between phase II and phase III) meet in the tricritical point (purple spot). The tricritical point lies at the intersection of the curves defined by $F'(0)=1$ (red line) and $F'''(0)=0$ (turquoise dashed line). The lines indicating the transitions separate the trivial phase I where $m=0$ from the phase II  where $m$ assumes one of the nonzero values $m_{+}$ or $m_{-}=-m_{+}$ (symmetry breaking), and from the phase III where the trivial solution coexists with solutions of broken symmetry.
	(b) Self-consistency map, cf. Eq.~\eqref{eq:selfconsistencymap}, with representatives for each of the phases. The fixed points are indicated by circles. Phase I: blue line with a single fixed point $m=0$ which is stable, $(a,b)=(-1.2, 2.5)$. Phase II: red line with three fixed points, the nonzero ones are stable, $(a,b)=(-0.5, 3.5)$. Phase III: green line with five fixed points, where $m=0$ and the two outermost ones are stable, $(a,b)=(-1.2, 2.5)$.
\label{fig:phasediagram}}
\end{figure}
On the other hand we prove in Appendix \ref{app:D} that
\begin{align}
	\langle x^{6}\rangle_{\text{tc}} - 15 \langle x^{2}\rangle^{3}_{\text{tc}} < 0.
	\label{eq:inequalitysixthmoment2}
\end{align}
Note that this is the same as \eqref{eq:momentrelation1} for $n=6$ albeit it holds for different reasons here.
These moment inequalities translate with Eqs.~(\ref{eq:tricriticalpointcondition}, \ref{eq:sixthmoment}) into bounds for the tricritical point, namely
\begin{align}
	9 \Big( \frac{\sigma^2}{2D} \Big)^2 <a_{\text{tc}}+ 3b_{\text{tc}} \frac{\sigma^2}{2D}  < 15 \Big( \frac{\sigma^2}{2D} \Big)^2.
\label{eq:tricritbounds}
\end{align}
These bounds are optimal in the sense that they are asymptotically assumed in the limits of strong and weak noise, which can be shown using a similar series ansatz as used in Appendix \ref{app:C}. In Fig.~\ref{fig:tricritical} we present numerical results for the tricritical point and compare it to its bounds for strong and weak noise.
\begin{figure}[h]
	\includegraphics{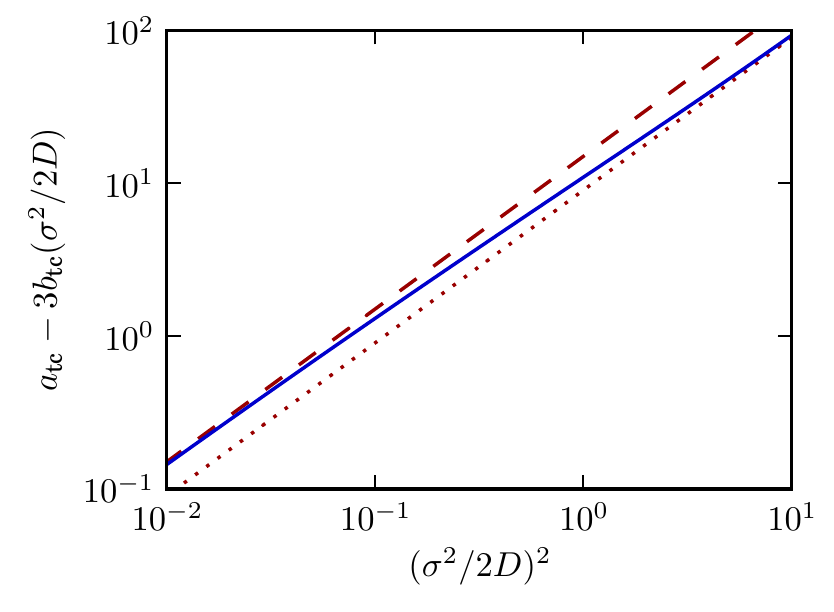}
	\caption{(color online) $a-3b \sigma^{2}/(2D)$ at the tricritical point $(a_{\text{tc}}, b_{\text{tc}})$ as a function of $(\sigma^{2}/2D)^{2}$ on a double logarithmic scale.
	Numerical results for $D=1$ and different noise strengths are displayed by the solid line (blue).
	Optimal bounds for the tricritical point are given in \eqref{eq:tricritbounds}. The upper bound (red dashed line) is asymptotically reached for weak noise, the lower bound (red dotted line) is asymptotically reached for strong noise.\label{fig:tricritical}}
\end{figure}
\subsection{Simulations in the region of coexistence}

Simulations are performed always for systems of finite size.
For finite systems the stationary solution $p_{\text{s}}(\mathbf{x})$ of the nonlinear Fokker-Planck equation \eqref{eq:nonlinearfokkerplanck} given by Eq.~\eqref{eq:statsolutiongeneral} is symmetric with respect to zero.
For the infinite system, as just discussed in detail, there is a breaking of symmetry.
There exists a coexistence region \footnote{
Coexistence in this system is not meant as coexistence of phases in the sense that part of the system is in one phase and part of the system is in another phase at the same time. Such phenomena are impossible in the present system. The physical reason for that is the global coupling. Since each coordinate is interacting with all others they relax to one common stationary distribution. Mathematically it can be understood by the nonlinear Fokker-Planck equation \eqref{eq:nonlinearfokkerplanck}. Due to the nonlinear term there is no superposition principle and the complete system can be only in one of the stationary solutions. That means by coexistence we refer to the coexistence of different solutions of the nonlinear Fokker-Planck equation and in particular solutions of very different character that can be symmetric or symmetry-breaking.
} in parameter space where two kinds of stationary solutions of the nonlinear Fokker-Planck equation \eqref{eq:nonlinearfokkerplanck} occur, a symmetric one and two symmetry-broken solutions.

We are interested in simulations of -necessarily- finite systems in the coexistence region of the infinite system.
It is useful to investigate the center of mass
\begin{align}
	R:= \frac{1}{L}\sum_{i=1}^{L} x_{i}
	\label{eq:centerofmass}
\end{align}
since in the limit of large system sizes it converges to the mean field
\begin{align}
	\lim_{L \rightarrow \infty} R = m.
	\label{eq:limitcom}
\end{align}
For finite systems $R$ is a stochastic variable with stationary distribution
\begin{align}
	p_{\text{s}}(R)= \int_{\mathbb{R}^{L}}^{} \diff \mathbf{x} \delta\Big(R-\frac{1}{L}\sum_{i=1}^{L} x_{i}\Big) p_{\text{s}}(\mathbf{x}).
	\label{eq:comdistr}
\end{align}
We expect that $R$ fluctuates around stable fixed points of the self-consistency-map \eqref{eq:selfconsistencymap} which means that in the region of coexistence the distribution has three peaks. It is interesting to investigate by simulations which of these peaks have the largest weights.

As a standard simulation needs to much time to reach the stationary distribution we use the method of patchwork sampling \cite{KB16} where the state space is cut into many patches that are then simulated separately. Eventually the results from the simulations of each piece are used to obtain the stationary distribution of the original model. In the present system the potential \eqref{eq:nonlinearsystem3} can have up two three minima that are separated by up to two potential barriers. Each particle might need a long time to overcome such a potential barrier. That is why the stationary distribution is assumed only after a very long time.
We consider the case of two local potential maxima at $x_{\text{max},+}>0$ and $x_{\text{max},-}<0$, with
$x_{\text{max}, +}=-x_{\text{max},-}$. Then we use the following partition of the state space into patches labeled by indices $k, l$. A configuration of the state space belongs to the patch $X_{kl}$ if there are exactly $k$ coordinates $x_i$ with $x_i<x_{\text{max},-}$ and exactly $l$ coordinates $x_{j}$ with $x_{j}>x_{\text{max}, +}$.

In Fig.~\ref{fig:patchwork} we present simulation results for the center of mass distribution for system size $L=200$ from the region of coexistence and slightly beyond. In the region of coexistence the distribution has three local maxima which correspond to the stable fixed points of the self-consistency map. However the weight of these three peaks can differ dramatically.
We find two typical scenarios. 

In the first one, presented in Fig.~\ref{fig:patchwork}a, the peaks around the two stable nonzero fixed points of the self-consistency map dominate the peak around zero.
Coming closer to the boundary towards the symmetry-broken phase of the infinite system (phase II), i.e. for decreasing $a$, the central peak looses weight and crossing the boundary it disappears completely.

In the second scenario, see Fig.~\ref{fig:patchwork}b, the central peak dominates the outer peaks.
Coming closer to the boundary towards the symmetric phase of the infinite system (phase I), i.e. for increasing $a$, the outer peaks loose weight and disappear crossing the boundary.

In Fig.~\ref{fig:patchwork2} we present the scaling of the center of mass distribution for different system sizes $L$ for the two scenarios.
In either case the weights of the suppressed peaks decrease at least exponentially with increasing $L$, see also Fig.~\ref{fig:patchwork3}.

We expect that in the limit $L\rightarrow \infty$ the weight of either the central peak (scenario one) or the outer peaks (scenario two) goes to zero. 
We conjecture that the region of coexistence can be divided into two parts separated by some critical line, where in one part scenario one is springing into action and in the other part scenario two.
It is an interesting topic of future research to calculate this critical line.

We further see on Fig.~\ref{fig:patchwork2} that all peaks, also the suppressed ones, become sharper for larger system sizes.
Hence locally the suppressed peak gains stability whereat it globally loses stability.
That means if we consider larger and larger systems that are prepared initially in the suppressed state, they will stay there longer and longer.
If we however wait infinite long each time, then the probability to find the system in the suppressed state will go to zero for large system sizes.

We also simulated the system sizes $L=25, 50, 100$ for the parameters belonging to the curves of Fig.~\ref{fig:patchwork}.
For parameters $a$ different from the ones of Fig.~\ref{fig:patchwork2} the decay of the suppressed peak is even much stronger.
In that sense the simulation results presented in Figs.~\ref{fig:patchwork2} and \ref{fig:patchwork3} are the most critical.
That means the parameters $a$ chosen in Fig.~\ref{fig:patchwork3} are the ones closest to the boundary separating scenario one from scenario two.

\begin{figure}[]
	\includegraphics{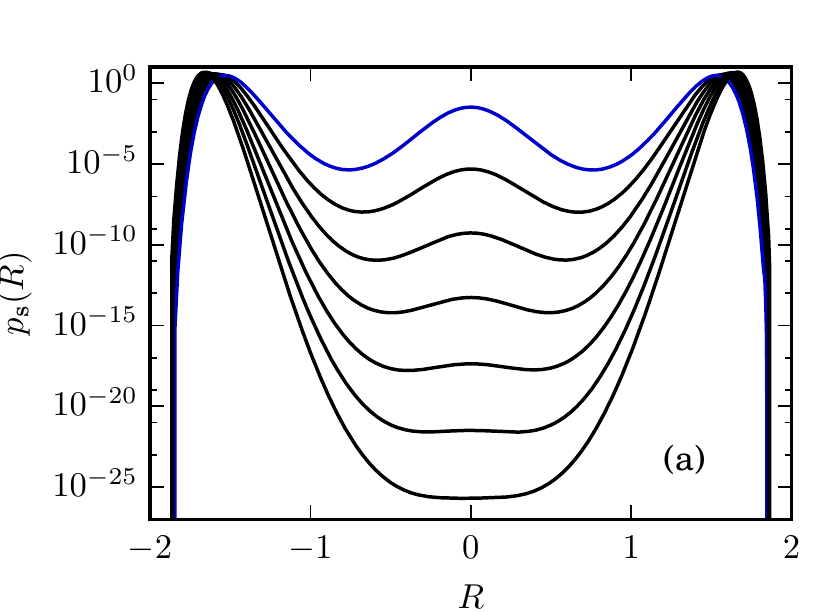}
	\includegraphics{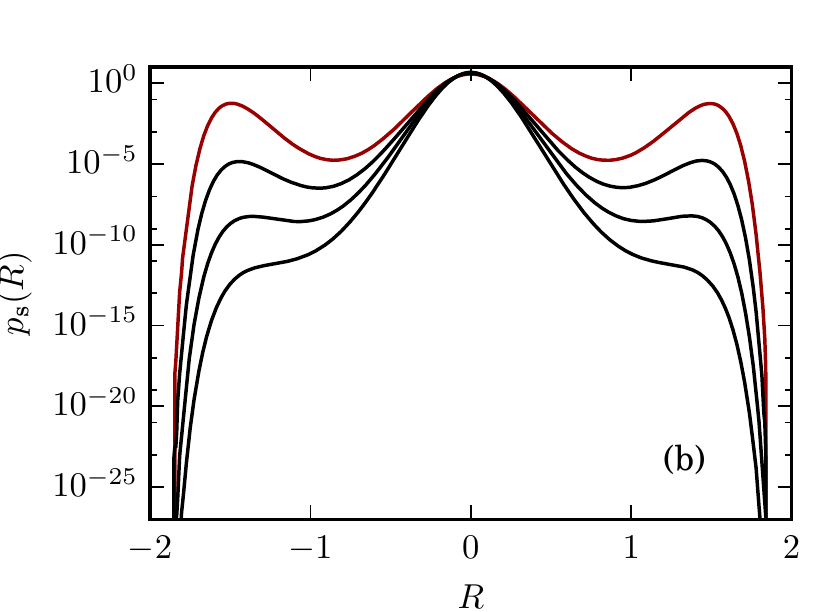}
	\caption{(color online) Stationary center of mass distribution $p_{\text{s}}(R)$ for $L=200$, $D=1$, $\sigma=1$, $b=4.5$ and (a) $a=-3.48, -3.50, \dots, -3.60$. The top solid (blue) line presents $a=-3.60$ and lower lines show higher values of $a$. System parameters are such that the infinite system is in phase III except for the lowest line where it is in phase II.
	(b) $a=-3.62, -3.64, -3.66, -3.68$. The top solid (red) line presents $a=-3.62$ and lower lines show smaller values of $a$. System parameters are such that the infinite system is in phase III except for the lowest line where it is in phase I. Simulations obtained via patchwork sampling with time step size $\Delta t = 10^{-3}$ and $10^{7}$ time steps for each patch.\label{fig:patchwork}}
\end{figure}
\begin{figure}[]
	\includegraphics{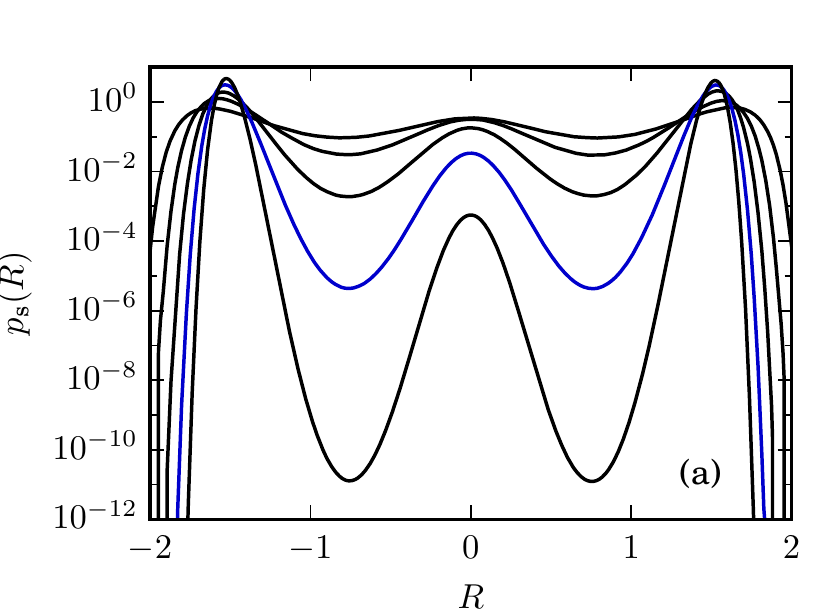}
	\includegraphics{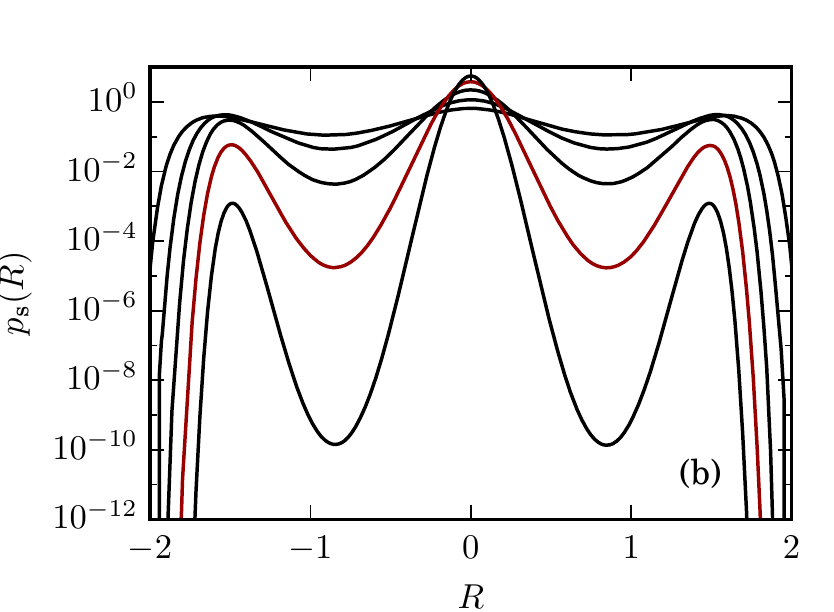}
	\caption{(color online) Stationary center of mass distribution $p_{\text{s}}(R)$ for $D=1$, $\sigma=1$, $b=4.5$ and $L=25, 50, 100, 200, 400$ for (a) $a=-3.60$ and (b) $a=-3.62$. The system size $L=200$ is colored (blue/red) and represents the same data as the colored curves in Fig.~\ref{fig:patchwork}. Upper lines at $R=1$ correspond to smaller system sizes and the lowest line corresponds to $L=400$. Simulations obtained via patchwork sampling with time step size $\Delta t = 10^{-3}$ and $10^{7}$ time steps for each patch.
\label{fig:patchwork2}}
\end{figure}
\begin{figure}
	\includegraphics{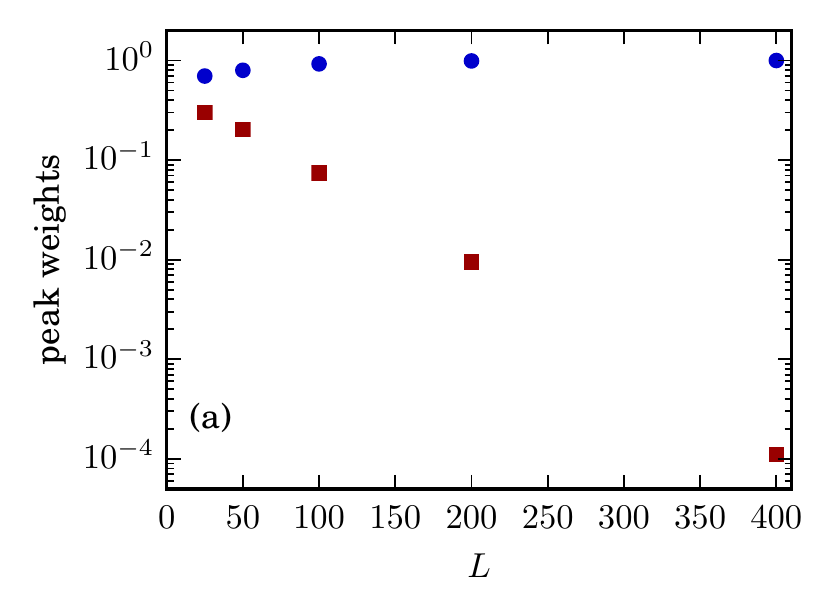}
	\includegraphics{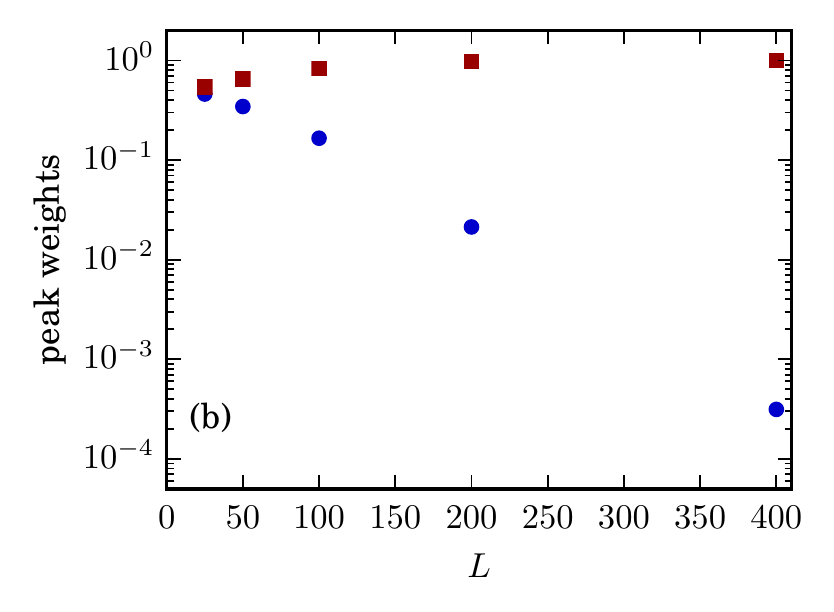}
	\caption{(color online) 
	Peak weights of the central peak (red squares) and the sum of the outer peaks (blue circles) of the center of mass distribution of Fig.~\ref{fig:patchwork2} in dependence on the system size. System parameters are $D=1$, $\sigma=1$, $b=4.5$ and (a) $a=-3.60$ and (b) $a=-3.62$. In (a) the weight of the central peak tends to zero at least exponentially and the sum of the weights of the outer peaks tends to one. In (b) the sum of the weights of the outer peaks tends to zero at least exponentially and the weight of the central peak tends to one.
\label{fig:patchwork3}}
\end{figure}

\section{Conclusions\label{sec:conclusions}}

We considered the overdamped motion of infinitely many globally coupled particles in an anharmonic potential under the influence of additive Gaussian white noise.
In the case of a double-well potential there is a continuous symmetry breaking phase transition.
We proved the existence and uniqueness of the critical point, i.e., there are no reentrance transitions.
We further proved optimal upper and lower bounds for the critical point that are assumed in the limits of weak and strong noise.

In case of a potential with three minima the phase space divides into three regions.
In phase I there exists only one symmetric solution.
In phase II there are two symmetry breaking solutions and in phase III symmetric and symmetry breaking solutions coexist.
All three phases meet in one point, the tricritical point.
We proved optimal bounds for this point that are assumed for weak and strong noise.

We further investigated the region of coexistence, phase III, by simulations of the system of finite size.
We sampled the stationary probability distribution of the center of mass.
In the limit of infinite system size the center of mass becomes the mean field which is deterministic.
In contrast, for finite systems the center of mass is stochastic.
We found that the mean field of each stationary solution of the infinite system corresponds to a local maximum of the center of mass distribution of finite but large systems.
Hence in the region of coexistence there are three local maxima and we investigated the weights of these three peaks.
Our simulation results strongly suggest that the weights of either the central peak or of the two outer peaks goes to zero at least exponentially when the system size goes to infinity.

Thus we found that in the region of coexistence the two limits of infinite system size and infinite observation time do not commute.
When the limit of infinite system size is performed first we find two symmetry-broken and one symmetric solution.
When on the other hand the limit of infinite observation time that means stationarity, is performed first the stationary center of mass distribution, depending on parameters, consists of either one peak at zero or two peaks located at the two stable fixed points of the mean field of the infinite system.

Naturally, the question arises which of the two limits is physically more relevant.
In a typical situation with very large system size the infinite system limit should be considered first since a stochastic switching between the three local maxima of the center of mass distribution is typically not occurring at all within the observation time and the behavior of the system depends highly on initial conditions.
If however the system size is not too large and the observation time is long enough to switch between the local maxima, the other order of limits can become relevant and, depending on parameters, one type of solution occurs in this case.
It is important to know the complete stationary distribution also when a system is perturbed by an external time dependent signal
which allows switching between the peaks; see, e.g., \cite{HNV12}.

\begin{acknowledgments}
R. K. thanks the International Max Planck Research School, Mathematics in the Sciences, Leipzig for supporting part of this work with a scholarship.
\end{acknowledgments}

\appendix

\section{Monotonicity of the critical point\label{app:A}}

We will show that for a double-well potential with saturation term of order $n\ge 6$, cf. Eq.~\eqref{eq:nonlinearsystem2}, the function $a_{c}(\sigma^{2}, D)$ is strictly monotone with respect to both arguments. At the critical point we always have $m=0$, so we drop this subscript here.

From Eq.~\eqref{eq:PTCmomentrelation} we have at the critical point $\langle x^{2}\rangle= \sigma^{2}/(2D)$. Differentiating with respect to $D$ gives
\begin{align}
	\frac{\diff}{\diff D} \langle x^{2}\rangle= - \frac{\sigma^2}{2D^{2}}= -\frac{2}{\sigma^{2}} \langle x^{2}\rangle.
	\label{eq:diffxsqD}
\end{align}
On the other hand considering $\langle x^{2}\rangle\big( a_{c}(\sigma^2, D), \sigma^2, D   \big)$ we have
\begin{align}
	\frac{\diff}{\diff D} \langle x^{2}\rangle = \partial_{D} \langle x^{2}\rangle + \partial_{a_{c}} \langle x^{2}\rangle \partial_{D} a_{c}(\sigma^2, D).
	\label{eq:diffxsqD2}
\end{align}
Using the explicit expression for $p_{\text{s}}(x,0)$, cf. Eqs.~(\ref{eq:statsolutionnonlinearfp},\ref{eq:statnormalization}), we calculate
\begin{align}
	\partial_{a_{c}}\langle x^{2}\rangle = - \partial_{D}\langle x^{2}\rangle= \frac{1}{\sigma^{2}} \Big(  \langle x^{4}\rangle - \langle x^{2}\rangle^{2} \Big).
	\label{eq:partialderivatievacD}
\end{align}
Inserting this into \eqref{eq:diffxsqD2} we obtain with \eqref{eq:diffxsqD}
\begin{align}
	\partial_{D} a_{c}(\sigma^{2}, D) = \frac{\langle x^{4}\rangle - 3 \langle x^{2}\rangle^{2}}{\langle x^{4}\rangle - \langle x^{2}\rangle^{2}} <0,
	\label{eq:partialDac}
\end{align}
since the numerator is negative as stated by Lemma \ref{lemma:3} of Appendix \ref{app:B} and the denominator is positive for any extended distribution.

In a similar way we compute the dependence of $\langle x^{2}\rangle$ on $\sigma^2$. It is convenient to substitute $s:=2/\sigma^{2}$, then
\begin{align}
	\frac{\diff}{\diff s}\langle x^{2}\rangle = - \frac{1}{Ds^{2}} = - D \langle x^{2}\rangle^{2}.
	\label{eq:diffxsqs}
\end{align}
On the other hand
\begin{align}
	\frac{\diff}{\diff s}\langle x^{2}\rangle = \partial_{s} \langle x^{2}\rangle + \partial_{a_{c}} \langle x^{2}\rangle \partial_{s} a_{c}(s,D).
	\label{eq:diffxsqs2}
\end{align}
Similar as above we calculate using Eqs.~(\ref{eq:statsolutionnonlinearfp},\ref{eq:statnormalization}) for $p_{\text{s}}(x,0)$
\begin{align}
	\partial_{s}\langle x^{2}\rangle=& \frac{a_{c}-D}{2}\big(  \langle x^{4}\rangle - \langle x^{2}\rangle^{2} \big)
	\notag
	\\
	&-\frac{1}{n} \big( \langle x^{n+2}\rangle - \langle x^{n}\rangle \langle x^{2}\rangle   \big).
	\label{eq:partialderivativesxsq}
\end{align}
To compute $\langle x^{n}\rangle$ and $\langle x^{n+2}\rangle$ we evaluate Eq.~\eqref{eq:relmomens2} for $k=0$ and $k=2$, observe $\sigma^{2}/(2D)= \langle x^{2}\rangle$ and arrive at 
\begin{align}
	\langle x^{n}\rangle &= a_{c}\langle x^{2}\rangle, 
	\\
	\langle x^{n+2}\rangle &= 3D \langle x^{2}\rangle^{2} + (a_{c} -D) \langle x^{4}\rangle.
	\label{eq:highermoments}
\end{align}
With these moments and Eqs.~(\ref{eq:partialderivatievacD},\ref{eq:diffxsqs}-\ref{eq:partialderivativesxsq}) we obtain
\begin{align}
	&\partial_{s}a_{c}(s,D)
	\label{eq:monotonicity2}
	\\
	&= \frac{ \big(\frac{2}{n} -1 \big) \Big[a_{c}\big(\langle x^{4}\rangle - \langle x^{2}\rangle^{2} \big)- D\big(\langle x^{4}\rangle - 3\langle x^{2}\rangle^{2} \big) \Big] }{s \big( \langle x^{4}\rangle - \langle x^{2}\rangle^{2} \big)}<0.
	\notag
\end{align}
The right hand side is negative since $2/n-1<0$, $\langle x^{4}\rangle - \langle x^{2}\rangle^{2} >0$ for any extended distribution, $\langle x^{4}\rangle - 3 \langle x^{2}\rangle^{2}<0$ as stated by Lemma \ref{lemma:3}, $D>0$ and $a_{c}>0$, cf. \eqref{eq:boundsonac}.
Thus $\partial_{\sigma^{2}} a_{c}(\sigma^2, D)>0$.
\section{Moment inequalities\label{app:B}}

We prove the moment inequalities (\ref{eq:momentrelation1},\ref{eq:momentrelation2}).
The proof of \eqref{eq:momentrelation1} is a consequence of the Cauchy-Schwarz inequality. If $n/2$ is even we have for any extended distribution
\begin{align}
	\langle x^{n/2}\rangle^{2} < \langle x^{n}\rangle
	\label{eq:cauchyschwarz1}
\end{align}
and if $n/2$ is odd
\begin{align}
	\langle x^{n/2+1}\rangle^{2}< \langle x^{2} \rangle \langle x^{n}\rangle.
	\label{eq:cauchyschwarz2}
\end{align}
For $n=4$ the claim \eqref{eq:momentrelation1} is directly given by \eqref{eq:cauchyschwarz1} and for $n> 4$ it follows by induction from \eqref{eq:cauchyschwarz1} or \eqref{eq:cauchyschwarz2}.

The proof of \eqref{eq:momentrelation2} uses a comparison between $p_{\text{s}}(x, m=0)$ and a Gaussian, for which the corresponding expression of the left-hand-side of \eqref{eq:momentrelation2} is zero. To show this we use the following statement:

\begin{lemma}\label{lemma:1}
	For any even probability distribution and for any fixed even $l\ge 4$ there exist coefficients $\alpha^{(l)}_{ij} \ge 0$ only depending on the second moment of the distribution, such that
\begin{align}
	\langle x^{l} \rangle - (l-1)!!\langle x^{2}\rangle^{l/2} = \kappa_{l} + \sum_{i=2}^{l-2}\sum_{j=4}^{l-2}\alpha^{(l)}_{ij} \langle x^i\rangle \kappa_{j} ,
	\label{eq:reprcumulant}
\end{align}
where $\kappa_{k}$ is the $k$th cumulant and the sums are meant to be zero if the lower bound on the summation index is larger than the upper bound.
\end{lemma}

\begin{proof}

By induction. Base case: for $l=4$ the claim holds as
\begin{align}
	\langle x^4\rangle -3 \langle x^2\rangle^2 = \kappa_4 .
	\label{eq:basecase}
\end{align}
For distributions with finite moments we have the general relation \cite{SO94}
\begin{align}
	\partial_{\kappa_{j}} \langle x^{n}\rangle = \binom{n}{j} \langle x^{n-j}\rangle	
	\label{eq:fromkendalls}
\end{align}
from which we conclude with
\begin{align}
	\langle x^{n}\rangle = \sum_{j=1}^{\infty} \kappa_{j} \partial_{\kappa_{j}}\langle x^{n}\rangle
	\label{eq:kendallsconsequence}
\end{align}
that, cf. also, e.g., \cite{dasgupta08}
\begin{align}
	\kappa_n = \langle x^n\rangle - \sum_{k=1}^{n-1} \binom{n-1}{k-1} \kappa_{k} \langle x^{n-k}\rangle.
	\label{eq:recursioncumulants}
\end{align}
Note that for even distributions $\kappa_{i}=0$ for all odd $i$. 
Hence we find for even $l\ge 6$
\begin{align}
	\kappa_{l} &=  \langle x^l\rangle - (l-1) \langle x^{2}\rangle \langle x^{l-2}\rangle - \sum_{k=4}^{l-1} \binom{l-1}{k-1}\kappa_k \langle x^{l-k}\rangle
	\label{eq:indstep}
\end{align}
By induction hypothesis Eq.~\eqref{eq:reprcumulant} holds if $l$ is replaced by $l-2$. Then, substituting $\langle x^{l-2}\rangle$ in Eq.~\eqref{eq:indstep} leads to
\begin{align}
	&\kappa_l - \langle x^l\rangle + (l-1)!!\langle x^2\rangle^{l/2} = -(l-1)\langle x^2\rangle\kappa_{l-2} \notag \\ 
	& - (l-1)\langle x^2\rangle\sum_{i=2}^{l-4} \sum_{j=4}^{l-4} \alpha^{(l-2)}_{ij}\langle x^i\rangle \kappa_j 
	-\sum_{k=4}^{l-1} \binom{l-1}{k-1}\kappa_k \langle x^{l-k}\rangle \notag \\
	&= -\sum_{i=2}^{l-2} \sum_{j=4}^{l-2} \alpha^{(l)}_{ij}\langle x^i\rangle\kappa_j.
	\label{eq:inductionresult}
\end{align}
Comparing coefficients we find for even $l\ge 8$ and $i \in (2, 4, \dots, l-4)$, $j \in (4, 6, \dots, l-4)$ the recursion
\begin{align}
	\alpha^{(l)}_{ij} &= \binom{l-1}{j-1}\delta_{i,l-j} + (l-1) \langle x^2\rangle \alpha^{(l-2)}_{ij} \ge 0.
	\label{eq:alphaexplicit2}
\end{align}
For even $l\ge 6$ we have
\begin{align}
	\alpha^{(l)}_{2, l-2} &= \frac{1}{2} l \left( l-1 \right) >0 .
	\label{eq:alphaexplicit1}
\end{align}
All other $\alpha^{(l)}_{ij}$ are zero. Thus the proof is complete.
\end{proof}
Although not needed here, the coefficients $\alpha_{ij}^{(l)}$ can be calculated as well. Eq.~\eqref{eq:alphaexplicit1} gives the initial value $\alpha^{(6)}_{2,4}=15$ for an iterative solution of Eq.~\eqref{eq:alphaexplicit2}.

Following a method of Dyson \cite{Dyson43} we introduce a Gaussian probability distribution $g(x)$ that is normalized and has the same first and second moment as $p_{\text{s}}(x, m=0)$. We aim to construct a polynomial that has the same sign as $g(x)-p_{\text{s}}(x, m=0)$ when nonzero. For the construction of the polynomial the intersection points of $g(x)$ and $p_{\text{s}}(x, m=0)$ are crucial as they are the points where $g(x)-p_{\text{s}}(x, m=0)$ changes its sign. The following Lemma states that two intersection points are not possible.
\begin{lemma}\label{lemma:2}
If $f$ and $g$ are two even, normalized, continuous probability distributions with equal variance, and there exists $a>0$ such that
\begin{align}
	g(x) \le f(x)&  \text{ for } |x| \in [-a, a], 
	\notag \\
	g(x) \ge f(x)&  \text{ else}, 
	\notag
\end{align}
then they are identical $f\equiv g$.
\end{lemma}
\begin{proof}
Consider the quadratic function $(x-a)(x+a)$. It has either the same sign as $g(x) - f(x)$ or at least one of them is zero. Hence
\begin{align}
	(x-a)(x+a)\big[g(x)-f(x)\big] \ge 0.
	\label{eq:cond1}
\end{align}
Assume there exists $x_0 \in \mathbb{R}$ with $f(x_0) \neq g(x_0)$. Then there exists an interval $(b,c)$ and $\varepsilon>0$ such that
\begin{align}
	(x-a)(x+a)\big[g(x)-f(x)\big] > \varepsilon
	\label{eq:cond2}
\end{align}
on $(b,c)$. Therefore
\begin{align}
	\int_{-\infty}^{\infty} \diff x \left( x^2 - a^2 \right)\big[g(x)-f(x)\big] > 0,
	\label{eq:ineqquadraticfunction}
\end{align}
which is a contradiction since $f$ and $g$ are normalized and have equal variance. Hence $f\equiv g$.
\end{proof}
Since $p_{\text{s}}(x, m=0)$ and $g(x)$ are both even they must have an even number of intersection points. Due to normalization zero intersection points are not possible. Since these distributions are not equal, due to Lemma~\ref{lemma:2}, also two intersection points are not possible. Six or more intersection points are not possible as well, since the intersection points are real solutions of the equation
\begin{align}
	p_{\text{s}}(x, m=0)=g(x).
	\label{eq:intersectionpoints}
\end{align}
Taking the exponential of both sides we find that the solutions are roots of the polynomial
\begin{align}
	\frac{1}{n}x^{n} + \beta x^{2} + \gamma,
	\label{eq:intersectionpointspolynomial}
\end{align}
where $\beta$ depends on the parameter $a$, the variance of $g(x)$, on normalization and on the noise strength, $\gamma$ is basically a normalization constant but also depends on the noise strength. The explicit expressions for $\beta$ and $\gamma$ can be given but are not important. From Descartes sign rule we find that there can not be more than two positive roots of the polynomial \eqref{eq:intersectionpointspolynomial} and due to symmetry there can not be more than two negative roots as well. Hence we find that there must be exactly four intersection points between $p_{\text{s}}(x, m=0)$ and $g(x)$ and we can apply the following Lemma.
\begin{lemma}\label{lemma:3}
If $f$ and $g$ are different, even, continuous and normalized probability distributions with equal variance and finite moments, and there exist $a_2>a_1>0$ such that
\begin{align}
	g(x) \le f(x)& \text{ for } |x| \in [a_{1}, a_{2}],
	\notag
	\\
	g(x) \ge f(x)& \text{ else},
	\notag 
\end{align}
then for all even $n \ge 4$ it holds
\begin{align}
	\langle x^n\rangle_{g} > \langle x^{n}\rangle_{f}.
	\label{eq:moments}
\end{align}
\end{lemma}
\begin{proof}
Consider the polynomials 
\begin{align}
q_k(x):=&x^{2k}(x-a_1)(x+a_1)(x-a_2)(x+a_2)
	\label{eq:pk} \\
	=& x^{2k+4} - (a_{1}^{2} + a_2^{2})x^{2k+2} +a_1^{2}a_{2}^{2} x^{2k},
\notag 
\end{align}
where $k=0,1,2,\dots$. These polynomials and the function $g(x)-f(x)$ have either the same sign or at least one of them is zero, hence
\begin{align}
	\int_{-\infty}^{\infty} \diff x q_k(x)\big[g(x)-f(x)\big] \ge 0.
	\label{eq:ineqpk1}
\end{align}
Since $f$ and $g$ are not identical and because of continuity we have even strict inequality
\begin{align}
	\int_{-\infty}^{\infty} \diff x q_k(x)\big[g(x)-f(x)\big] > 0.
	\label{eq:ineqpk2}
\end{align}
For $n=4$ the claim follows directly from the case $k=0$ as 
\begin{align}
	\int_{-\infty}^{\infty}\diff x q_{0}(x) \left(g(x)-f(x)\right) = \langle x^4\rangle_{g} -\langle x^4\rangle_{f} >0.
	\label{eq:indbasecase}
\end{align}
Note that the zeroth and second power of $x$ vanish in the above integral, since $f$ and $g$ are normalized and have equal variance. For even $n \ge 6$ consider the polynomial
\begin{align}
	Q_n(x) := \sum_{k=0}^{n/2-2} a_{2}^{n-4-2k}q_{k}(x) = \sum_{k=0}^{n/2} b_{2k}x^{2k}.
	\label{eq:polk}
\end{align}
One easily checks that 
\begin{align}
	b_{0}&= a_{1}^{2}a_{2}^{n-2},
	\notag
	\\
	b_{2} &= -a_{2}^{n-2},
	\notag
	\\
	b_{n-2} &= -a_{1}^{2},
	\notag
	\\
	b_{n} &= 1,
	\label{eq:coefficientspol}
\end{align}
and all other $b_{k}$ are zero.
Multiplication with $g-f$ and integration yields
\begin{align}
	0<&\int_{-\infty}^{\infty}\diff x  Q_n(x)\big[g(x)-f(x)\big] 
	\notag \\
	=&  \int_{-\infty}^{\infty} \diff x \sum_{k=2}^{n/2} b_{2k}x^{2k}\big[g(x)-f(x)\big].
	\label{eq:recursion1}
\end{align}
Inserting the coefficients \eqref{eq:coefficientspol} into inequality \eqref{eq:recursion1} we find
\begin{align}
	\langle x^{n}\rangle_{g} - \langle x^{n}\rangle_{f} > a_{1}^{2}(\langle x^{n-2}\rangle_{g}-\langle x^{n-2}\rangle_{f}),
	\label{eq:inductionstep}
\end{align}
where we used again that the zeroth and the second moment are equal for $f$ and $g$.
Thus the claim follows by induction.
\end{proof}
Taking $p_{\text{s}}(x, m=0)$ as $f(x)$ and a Gaussian with the same variance and mean as $g(x)$, Lemma \ref{lemma:3} can be applied since $p_{\text{s}}(x, m=0)$ is smaller than $g(x)$ for large enough $|x|$ and since they intersect in exactly four points. Hence
\begin{align}
	\langle x^{n}\rangle_{p_{\text{s}}}<\langle x^{n}\rangle_{g}.
	\label{eq:inequalityhighestmoment}
\end{align}
For a Gaussian all cumulants of higher than second order are zero. Hence due to Lemma~\ref{lemma:1}
\begin{align}
	\langle x^{n} \rangle_{g} - (n-1)!!\langle x^{2}\rangle_{g}^{n/2} = 0.
	\label{eq:momentrelationgaussian}
\end{align}
With $\langle x^{2}\rangle_{p_{\text{s}}}= \langle x^{2}\rangle_{g}$ and \eqref{eq:inequalityhighestmoment} it follows \eqref{eq:momentrelation2}.
By the same argument  we conclude that 
\begin{align}
	\kappa_{4} = \langle x^{4}\rangle - 3 \langle x^{2}\rangle^{2} <0.
	\label{eq:kurtosis}
\end{align}
\section{Strong and weak noise\label{app:C}}

We show that the upper bound in inequality \eqref{eq:boundsonac} is assumed asymptotically for weak noise $\sigma^2 \rightarrow 0$ and the lower bound is assumed asymptotically for strong noise $\sigma^2 \rightarrow \infty$.

For weak noise we use again the substitution $s:=2/\sigma^2$. From inequality \eqref{eq:boundsonac} we observe that $a_c \rightarrow 0$ as $\sigma^2 \rightarrow 0$. Therefore we make the ansatz
\begin{align}
	a_c(s, D) = a_{1}(D) \frac{1}{s} + a_2(D)\frac{1}{s^2}+\dots .
	\label{eq:ansatzacweaknoise}
\end{align}
The phase transition condition Eq.~\eqref{eq:PTCmomentrelation} is equivalent to
\begin{align}
	1&= 2D s \partial_{a_{1}} \ln I,
	\label{eq:ptcweaknoise}
	\\
	I&= \int_{-\infty}^{\infty} \diff x \exp\bigg[ -s\Big( \frac{D-a_c}{2} x^2 +\frac{1}{n} x^n \Big)   \bigg].
	\label{eq:laplace1}
\end{align}
From inequality \eqref{eq:boundsonac} we find that for fixed $D$ and large enough $s$ we have $D-a_c>0$. Hence we apply Laplace's method around $x=0$ to evaluate the integral \eqref{eq:laplace1}
\begin{align}
	I=& \int_{-\infty}^{\infty} \diff x \exp\Big( -s \frac{D-a_c}{2} x^2 \Big) \exp \Big( - \frac{s}{n} x^n \Big)
	\notag
	\\
	=&\int_{-\infty}^{\infty} \diff x \exp\Big( -s \frac{D-a_c}{2} x^2 \Big) \times \Big(1 - \frac{s}{n} x^n + \mathcal{O}(s^2 x^{2n})   \Big) 
	\notag
	\\
	=& \sqrt{\frac{2 \pi}{s(D-a_c)}} 
	\notag
	\\
	&\times \bigg[1 + \frac{(n-1)!!}{n} s^{1-n/2} \Big(\frac{1}{D-a_c} \Big)^{n/2} + \mathcal{O}(s^{2-n}) \bigg].
	\label{eq:laplace2}
\end{align}
Inserting the ansatz \eqref{eq:ansatzacweaknoise} we find
\begin{align}
	\partial_{a_{c}}\ln I =& \frac{1}{2s(D-a_c)}
	+ \bigg[\frac{(n-1)!!}{2}s^{-n/2} \Big(\frac{1}{D-a_{c}} \Big)^{1+ n/2} \bigg]
	\notag
	\\
	& \Big/ \bigg[1+ \frac{(n-1)!!}{n}s^{1-n/2} \Big(\frac{1}{D-a_{c}} \Big)^{n/2}\bigg] + \mathcal{O}(s^{-n/2})
	\notag
	\\
	=&\frac{1}{2s(D-a_c)} + \frac{1}{2D}(n-1)!! (sD)^{-n/2} 
	\notag
	\\
	&+ \mathcal{O}(s^{-n/2}).
	\label{eq:preptc}
\end{align}
Inserting this expression into the phase transition condition Eq.~\eqref{eq:ptcweaknoise} we obtain
\begin{align}
	1 = \frac{1}{1 - \frac{1}{D}a_{c}} + (n-1)!! (sD)^{1-n/2} \frac{1}{D} + \mathcal{O}(s^{-n/2}).
	\label{eq:ptcweaknoiseleading}
\end{align}
Inserting the ansatz \eqref{eq:ansatzacweaknoise} for $a_c$ and comparing coefficients at different powers of $s$ we find
\begin{align}
	a_{i}&=0 \text{ for } i=1, \dots, n/2-2,
	\notag
	\\
	a_{n/2-1}&= (n-1)!! D^{1-n/2}.
	\label{eq:leadingacweaknoise1}
\end{align}
Inserting these coefficients in the ansatz \eqref{eq:ansatzacweaknoise} we obtain the leading behavior of $a_c$ for weak noise
\begin{align}
	a_{c} \approx (n-1)!! \Big( \frac{\sigma^{2}}{2D}\Big)^{n/2-1},
	\label{eq:acweaknoise}
\end{align}
which is the upper bound for $a_c$ in \eqref{eq:boundsonac}.

For strong noise, due to the bounds \eqref{eq:boundsonac}, we use the ansatz
\begin{align}
	a_{c} = a_{1}(D) (\sigma^{2})^{n/2-1} + a_{2}(D) (\sigma^{2})^{n/2-2} +  \dots .
	\label{eq:ansatzacstrongnoise}
\end{align}
with the substitution
\begin{align}
	\lambda= \sigma^{n-2}
	\label{eq:substitutionstrongnoise}
\end{align}
the phase transition condition \eqref{eq:PTCmomentrelation} becomes
\begin{align}
	1=\frac{2D}{\lambda}\partial_{a_{1}} \ln I.
	\label{eq:ptcstrongnoise}
\end{align}
We use the substitution $x=\sigma y$ and the symmetry of the integrand to rewrite the integral \eqref{eq:laplace2} as
\begin{align}
	I&= 2\sigma \int_{0}^{\infty} \diff y \exp \big[ - \lambda \Phi(y)  \big],
	\notag
	\\
	\Phi(y)&= -\frac{a_{c}-D}{\lambda}y^2 + \frac{2}{n}y^{n}.
	\label{eq:subsintegral}
\end{align}
According to the bounds \eqref{eq:boundsonac}, for fixed $D$ and large enough noise we have $a_c-D>0$ and hence the integrand has its maximum at
\begin{align}
	y_{0}= \Big( \frac{a_{c}-D}{\lambda} \Big)^{1/(n-2)}.
	\label{eq:maximumintegrand}
\end{align}
Expanding $\Phi(y)$ around $y_0$ yields
\begin{align}
	I=& 2\sigma \exp\big[-\lambda \Phi(y_0)\big]
	\notag
	\\
	&\times \int_{0}^{\infty}\exp\big[-\lambda \Phi^{(2)}(y_0)/2 (y-y_0)^{2}\big] 
	\notag
	\\
	&\times\exp\Big[  -\lambda\big(\Phi^{(3)}(y_0)/6 (y-y_0)^{3} + \dots  \big) \Big] \diff y.
	\label{eq:integral3}
\end{align}
As $\lambda\rightarrow \infty$ the maximum becomes sharp and the main contribution to the integral comes only from the vicinity of the maximum. Hence we can change the lower bound to $-\infty$ without changing the behavior of the integral in the limit $\lambda \rightarrow \infty$. Expanding the last factor in the integrand around $y_0$ we obtain
\begin{align}
	I=& 2\sigma \exp\big[-\lambda \Phi(y_0)\big]
	\notag
	\\
	&\times \int_{0}^{\infty}\exp\big[-\lambda \Phi^{(2)}(y_0)/2 (y-y_0)^{2}\big] \big( 1 + \dots \big) \diff y.
	\label{eq:integral4}
\end{align}
Performing the Gaussian integrals we obtain
\begin{align}
	I= 2\sigma \exp\big[-\lambda \Phi(y_0)\big] \sqrt{\frac{2\pi}{\lambda \Phi^{(2)}(y_0)}} \Big[1+ \mathcal{O}(\lambda^{-1}) \Big].
	\label{eq:integral5}
\end{align}
Inserting this expression into Eq.~\eqref{eq:ptcstrongnoise} and using
\begin{align}
	\Phi(y_0)=& - \frac{n-2}{n} \Big(\frac{a_c-D}{\lambda} \Big)^{n/(n-2)},
	\notag
	\\
\Phi^{(2)}(y_0)=&2\bigg[-\frac{a_c -D}{\lambda}+ \Big( \frac{a_c -D}{\lambda}\Big)^{n/(n-2)}  \bigg],
\end{align}
and the ansatz \eqref{eq:ansatzacstrongnoise} we obtain
\begin{align}
	1 =& 2D \Big(\frac{a_c-D}{\lambda} \Big)^{2/(n-2)} 
	\notag
	\\
	&- \frac{2D}{\lambda} \frac{1- \frac{n(n-1)}{n-2} \Big(\frac{a_c-D}{\lambda} \Big)^{2/(n-2)}}{\frac{a_c-D}{\lambda} - (n-1) \Big(\frac{a_c-D}{\lambda} \Big)^{n/(n-2)}} + \mathcal{O}(\lambda^{-1}) 
	\notag
	\\
	&= 2D a_{1}^{2/(n-2)} + \mathcal{O}(\sigma^{-2}).
	\label{eq:ptcstrongnoise2}
\end{align}
Hence comparing coefficients of highest order in $\sigma^{2}$ yields
\begin{align}
	a_{1}= \Big(\frac{1}{2D} \Big)^{n/2-1}
	\label{eq:aonestrongnoise}
\end{align}
and thus the leading behavior of $a_c$ for strong noise is
\begin{align}
	a_{c}\approx \Big(\frac{\sigma^{2}}{2D} \Big)^{n/2-1},
	\label{eq:asympacstrongnoise}
\end{align}
which is the lower bound of $a_c$ in \eqref{eq:boundsonac}.

\section{Tricritical point inequalities\label{app:D}}

We prove that for $m=0$ the sixth cumulant $\kappa_{6}$ of the distribution $p_{\text{s}}(x, m=0)$, cf. Eq.~\eqref{eq:statsolutionnonlinearfp}, for the potential \eqref{eq:nonlinearsystem3} is negative when the fourth cumulant $\kappa_{4}$ is zero. Therefore we use the following Lemma. 
\begin{lemma}\label{lemma:sixthmoment}
	Consider a zero mean Gaussian probability distribution $g(x)$ and a smooth even probability distribution $p(x)$ with the same variance and fourth moment. Let $p(x)$ and $g(x)$ be not identical and not intersect in more than six points. Let further $g(x)>p(x)$ for large enough $x$.

	Then $g(x)$ and $p(x)$ intersect in exactly six points and 
	\begin{align}
		\langle x^{6}\rangle_{g} > \langle x^{6}\rangle_{p},
		\label{eq:sixthmomentlemma}
	\end{align}
	where $\langle x^{6}\rangle_{g}$ and $\langle x^{6}\rangle_{p}$ denote the sixth moment of $g(x)$ and $p(x)$, respectively.
\end{lemma}
\begin{proof}
	The proof is based on an idea of Dyson \cite{Dyson43}.
Since both distributions are normalized and not identical they must intersect at least in one point.
The number of intersection points must be even because both distributions are even.
Hence there can be two, four or six intersection points.

Assume there are exactly four intersection points $c_1$, $-c_1$, $c_2$ and $-c_2$. Then the polynomial
\begin{align}
	q_4(x):= (x-c_1)(x+c_1)(x-c_2)(x+c_2)
	\label{eq:poly4app4}
\end{align}
has always the same sign as $g(x)-p(x)$ except when it is zero. Therefore
\begin{align}
	\int_{-\infty}^{\infty} \diff x \big[ g(x)-p(x) \big] q_{4}(x) >0.
	\label{eq:int1app4}
\end{align}
But the polynomial $q_4$ is only of fourth order and $g(x)$ and $p(x)$ agree in the first four moments. Therefore the integral \eqref{eq:int1app4} must be zero which is a contradiction.
For two intersection points an analog contradiction can be constructed.
Hence there must be exactly six intersection points. Let them be $\pm c_i$ for $i=1,2,3$ and consider the polynomial
\begin{align}
	q_6(x):= \prod_{i=1}^{3} (x-c_i)(x+c_i).
	\label{eq:poly6app4}
\end{align}
Again this polynomial has the same sign as $g(x)-p(x)$ except where it is zero and thus
\begin{align}
	\int_{-\infty}^{\infty} \diff x \big[ g(x)-p(x) \big] q_{6}(x) >0.
	\label{eq:int2app4}
\end{align}
On the other hand, $g(x)$ and $p(x)$ are both even and therefore all odd moments vanish, they also agree in the second and the fourth moment. Hence they coincide in the first five moments.
Thus the only remaining term in Eq.~\eqref{eq:int2app4} is
\begin{align}
	\int_{-\infty}^{\infty} \diff x \big[ g(x)-p(x) \big] q_{6}(x) = \langle x^{6}\rangle_{g} - \langle x^{6}\rangle_{p}>0
	\label{eq:int3app4}
\end{align}
which proofs the claim \eqref{eq:sixthmomentlemma}.
\end{proof}

We consider the stationary  distribution $p_{\text{s}}(x, m=0)$ as $p$ and a zero mean Gaussian with the same variance as $g$ of Lemma~\ref{lemma:sixthmoment}. 
As for a Gaussian all cumulants starting from the third are zero also the fourth cumulant is zero.
Since the fourth cumulant of $p$ is zero as well $p$ and $g$ must coincide in the fourth moment.
Hence Lemma~\ref{lemma:sixthmoment} can be applied and thus $\langle x^{6}\rangle_{g}>\langle x^{6}\rangle_{p_{s}}$. 
As $g$ and $p_{\text{s}}$ agree in the first five moments the sixth cumulant of $p_{\text{s}}$ must be smaller than the one of $g$.
But as $g$ is Gaussian its sixth cumulant is zero and hence the sixth cumulant of $p_{\text{s}}$ is negative.
As the cumulants of $p_{\text{s}}$ depend continuously on parameters $\kappa_6$ must be negative also in a neighborhood of the parameter set where $\kappa_{4}=0$.
Thus \eqref{eq:negativesixthcumulant} is proved.

Furthermore we have for the zero mean Gaussian $g$
\begin{align}
	\langle x^{6}\rangle_{g} - 15 \langle x^{2}\rangle^{3}_{g}=\kappa_{6} + 15 \langle x^{2}\rangle\kappa_{4}=0.
	\label{eq:cumulantconstr}
\end{align}
As $p_{\text{s}}$ and $g$ agree in the first five moments it follows by Lemma~\ref{lemma:sixthmoment} and the same argumentation as above that for $p_{\text{s}}$ at the tricritical point, where $\kappa_4=0$, 
\begin{align}
	\langle x^{6}\rangle_{p_{s}} - 15 \langle x^{2}\rangle^{3}_{p_{s}}<0.
	\label{eq:cumulantconstr2}
\end{align}
Thus \eqref{eq:inequalitysixthmoment2} is proved.

\bibliography{literatur.bib}

\begin{thebibliography}{30}
\expandafter\ifx\csname natexlab\endcsname\relax\def\natexlab#1{#1}\fi
\expandafter\ifx\csname bibnamefont\endcsname\relax
  \def\bibnamefont#1{#1}\fi
\expandafter\ifx\csname bibfnamefont\endcsname\relax
  \def\bibfnamefont#1{#1}\fi
\expandafter\ifx\csname citenamefont\endcsname\relax
  \def\citenamefont#1{#1}\fi
\expandafter\ifx\csname url\endcsname\relax
  \def\url#1{\texttt{#1}}\fi
\expandafter\ifx\csname urlprefix\endcsname\relax\def\urlprefix{URL }\fi
\providecommand{\bibinfo}[2]{#2}
\providecommand{\eprint}[2][]{\url{#2}}

\bibitem[{\citenamefont{Gammaitoni et~al.}(1998)\citenamefont{Gammaitoni,
  H{\"a}nggi, Jung, and Marchesoni}}]{GHJM98}
\bibinfo{author}{\bibfnamefont{L.}~\bibnamefont{Gammaitoni}},
  \bibinfo{author}{\bibfnamefont{P.}~\bibnamefont{H{\"a}nggi}},
  \bibinfo{author}{\bibfnamefont{P.}~\bibnamefont{Jung}}, \bibnamefont{and}
  \bibinfo{author}{\bibfnamefont{F.}~\bibnamefont{Marchesoni}},
  \bibinfo{journal}{Rev. Mod. Phys.} \textbf{\bibinfo{volume}{70}},
  \bibinfo{pages}{223} (\bibinfo{year}{1998}).

\bibitem[{\citenamefont{Reimann}(2002)}]{Reimann02}
\bibinfo{author}{\bibfnamefont{P.}~\bibnamefont{Reimann}},
  \bibinfo{journal}{Phys. Rep.} \textbf{\bibinfo{volume}{361}},
  \bibinfo{pages}{57} (\bibinfo{year}{2002}).

\bibitem[{\citenamefont{Horsthemke and Lefever}(1984)}]{HL84}
\bibinfo{author}{\bibfnamefont{W.}~\bibnamefont{Horsthemke}} \bibnamefont{and}
  \bibinfo{author}{\bibfnamefont{R.}~\bibnamefont{Lefever}},
  \emph{\bibinfo{title}{Noise-induced transitions}}
  (\bibinfo{publisher}{Springer}, \bibinfo{address}{Berlin},
  \bibinfo{year}{1984}).

\bibitem[{\citenamefont{Sagu\'es et~al.}(2007)\citenamefont{Sagu\'es, Sancho,
  and Garc\'ia-Ojalvo}}]{SSG07}
\bibinfo{author}{\bibfnamefont{F.}~\bibnamefont{Sagu\'es}},
  \bibinfo{author}{\bibfnamefont{J.~M.} \bibnamefont{Sancho}},
  \bibnamefont{and}
  \bibinfo{author}{\bibfnamefont{J.}~\bibnamefont{Garc\'ia-Ojalvo}},
  \bibinfo{journal}{Rev. Mod. Phys.} \textbf{\bibinfo{volume}{79}},
  \bibinfo{pages}{829} (\bibinfo{year}{2007}).

\bibitem[{\citenamefont{García-Ojalvo and Sancho}(1999)}]{GS99}
\bibinfo{author}{\bibfnamefont{J.}~\bibnamefont{García-Ojalvo}}
  \bibnamefont{and} \bibinfo{author}{\bibfnamefont{J.~M.}
  \bibnamefont{Sancho}}, \emph{\bibinfo{title}{Noise in Spatially Extended
  Systems}} (\bibinfo{publisher}{Springer-Verlag}, \bibinfo{address}{New York,
  Berlin, Heidelberg}, \bibinfo{year}{1999}).

\bibitem[{\citenamefont{Frank}(2005)}]{Frank05}
\bibinfo{author}{\bibfnamefont{T.~D.} \bibnamefont{Frank}},
  \emph{\bibinfo{title}{Nonlinear Fokker-Planck Equations}}
  (\bibinfo{publisher}{Springer-Verlag}, \bibinfo{address}{Berlin, Heidelberg},
  \bibinfo{year}{2005}).

\bibitem[{\citenamefont{Shiino}(1985)}]{Shiino85}
\bibinfo{author}{\bibfnamefont{M.}~\bibnamefont{Shiino}},
  \bibinfo{journal}{Phys. Lett. A} \textbf{\bibinfo{volume}{112}},
  \bibinfo{pages}{302} (\bibinfo{year}{1985}).

\bibitem[{\citenamefont{Desai and Zwanzig}(1978)}]{DZ78}
\bibinfo{author}{\bibfnamefont{R.}~\bibnamefont{Desai}} \bibnamefont{and}
  \bibinfo{author}{\bibfnamefont{R.}~\bibnamefont{Zwanzig}},
  \bibinfo{journal}{J. Stat. Phys.} \textbf{\bibinfo{volume}{19}},
  \bibinfo{pages}{1} (\bibinfo{year}{1978}).

\bibitem[{\citenamefont{Garc{\'\i}a-Ojalvo
  et~al.}(1993)\citenamefont{Garc{\'\i}a-Ojalvo, Hern{\'a}ndez-Machado, and
  Sancho}}]{GHS93}
\bibinfo{author}{\bibfnamefont{J.}~\bibnamefont{Garc{\'\i}a-Ojalvo}},
  \bibinfo{author}{\bibfnamefont{A.}~\bibnamefont{Hern{\'a}ndez-Machado}},
  \bibnamefont{and} \bibinfo{author}{\bibfnamefont{J.~M.}
  \bibnamefont{Sancho}}, \bibinfo{journal}{Phys. Rev. Lett.}
  \textbf{\bibinfo{volume}{71}}, \bibinfo{pages}{1542} (\bibinfo{year}{1993}).

\bibitem[{\citenamefont{Birner et~al.}(2002)\citenamefont{Birner, Lippert,
  M\"uller, K\"uhnel, and Behn}}]{BLMKB02}
\bibinfo{author}{\bibfnamefont{T.}~\bibnamefont{Birner}},
  \bibinfo{author}{\bibfnamefont{K.}~\bibnamefont{Lippert}},
  \bibinfo{author}{\bibfnamefont{R.}~\bibnamefont{M\"uller}},
  \bibinfo{author}{\bibfnamefont{A.}~\bibnamefont{K\"uhnel}}, \bibnamefont{and}
  \bibinfo{author}{\bibfnamefont{U.}~\bibnamefont{Behn}},
  \bibinfo{journal}{Phys. Rev. E} \textbf{\bibinfo{volume}{65}},
  \bibinfo{pages}{046110} (\bibinfo{year}{2002}).

\bibitem[{\citenamefont{Van~den Broeck
  et~al.}(1994{\natexlab{a}})\citenamefont{Van~den Broeck, Parrondo, Armero,
  and Hern\'andez-Machado}}]{BPAH94}
\bibinfo{author}{\bibfnamefont{C.}~\bibnamefont{Van~den Broeck}},
  \bibinfo{author}{\bibfnamefont{J.~M.~R.} \bibnamefont{Parrondo}},
  \bibinfo{author}{\bibfnamefont{J.}~\bibnamefont{Armero}}, \bibnamefont{and}
  \bibinfo{author}{\bibfnamefont{A.}~\bibnamefont{Hern\'andez-Machado}},
  \bibinfo{journal}{Phys. Rev. E} \textbf{\bibinfo{volume}{49}},
  \bibinfo{pages}{2639} (\bibinfo{year}{1994}{\natexlab{a}}).

\bibitem[{\citenamefont{García-Ojalvo
  et~al.}(1996)\citenamefont{García-Ojalvo, Parrondo, Sancho, and Van~den
  Broeck}}]{GPSB96}
\bibinfo{author}{\bibfnamefont{J.}~\bibnamefont{García-Ojalvo}},
  \bibinfo{author}{\bibfnamefont{J.~M.~R.} \bibnamefont{Parrondo}},
  \bibinfo{author}{\bibfnamefont{J.~M.} \bibnamefont{Sancho}},
  \bibnamefont{and} \bibinfo{author}{\bibfnamefont{C.}~\bibnamefont{Van~den
  Broeck}}, \bibinfo{journal}{Phys. Rev. E} \textbf{\bibinfo{volume}{54}},
  \bibinfo{pages}{6918} (\bibinfo{year}{1996}).

\bibitem[{\citenamefont{Kometani and Shimizu}(1975)}]{KS75}
\bibinfo{author}{\bibfnamefont{K.}~\bibnamefont{Kometani}} \bibnamefont{and}
  \bibinfo{author}{\bibfnamefont{H.}~\bibnamefont{Shimizu}},
  \bibinfo{journal}{J. Stat. Phys.} \textbf{\bibinfo{volume}{13}},
  \bibinfo{pages}{473} (\bibinfo{year}{1975}).

\bibitem[{\citenamefont{Dawson}(1983)}]{Dawson83}
\bibinfo{author}{\bibfnamefont{D.~A.} \bibnamefont{Dawson}},
  \bibinfo{journal}{J. Stat. Phys.} \textbf{\bibinfo{volume}{31}},
  \bibinfo{pages}{29} (\bibinfo{year}{1983}).

\bibitem[{\citenamefont{Shiino}(1987)}]{Shiino87}
\bibinfo{author}{\bibfnamefont{M.}~\bibnamefont{Shiino}},
  \bibinfo{journal}{Phys. Rev. A} \textbf{\bibinfo{volume}{36}},
  \bibinfo{pages}{2393} (\bibinfo{year}{1987}).

\bibitem[{\citenamefont{K\"ursten et~al.}(2013)\citenamefont{K\"ursten,
  G\"utter, and Behn}}]{KGB13}
\bibinfo{author}{\bibfnamefont{R.}~\bibnamefont{K\"ursten}},
  \bibinfo{author}{\bibfnamefont{S.}~\bibnamefont{G\"utter}}, \bibnamefont{and}
  \bibinfo{author}{\bibfnamefont{U.}~\bibnamefont{Behn}},
  \bibinfo{journal}{Phys. Rev. E} \textbf{\bibinfo{volume}{88}},
  \bibinfo{pages}{022114} (\bibinfo{year}{2013}).

\bibitem[{\citenamefont{Van~den Broeck
  et~al.}(1994{\natexlab{b}})\citenamefont{Van~den Broeck, Parrondo, and
  Toral}}]{BPT94}
\bibinfo{author}{\bibfnamefont{C.}~\bibnamefont{Van~den Broeck}},
  \bibinfo{author}{\bibfnamefont{J.~M.~R.} \bibnamefont{Parrondo}},
  \bibnamefont{and} \bibinfo{author}{\bibfnamefont{R.}~\bibnamefont{Toral}},
  \bibinfo{journal}{Phys. Rev. Lett.} \textbf{\bibinfo{volume}{73}},
  \bibinfo{pages}{3395} (\bibinfo{year}{1994}{\natexlab{b}}).

\bibitem[{\citenamefont{Landa et~al.}(1998)\citenamefont{Landa, Zaikin, and
  Schimansky-Geier}}]{LZS98}
\bibinfo{author}{\bibfnamefont{P.~S.} \bibnamefont{Landa}},
  \bibinfo{author}{\bibfnamefont{A.}~\bibnamefont{Zaikin}}, \bibnamefont{and}
  \bibinfo{author}{\bibfnamefont{L.}~\bibnamefont{Schimansky-Geier}},
  \bibinfo{journal}{Chaos Soliton. Fract.} \textbf{\bibinfo{volume}{9}},
  \bibinfo{pages}{1367} (\bibinfo{year}{1998}).

\bibitem[{\citenamefont{Zaikin et~al.}(1999)\citenamefont{Zaikin,
  Garcia-Ojalvo, and Schimansky-Geier}}]{ZGS99}
\bibinfo{author}{\bibfnamefont{A.}~\bibnamefont{Zaikin}},
  \bibinfo{author}{\bibfnamefont{J.}~\bibnamefont{Garcia-Ojalvo}},
  \bibnamefont{and}
  \bibinfo{author}{\bibfnamefont{L.}~\bibnamefont{Schimansky-Geier}},
  \bibinfo{journal}{Phys. Rev. E} \textbf{\bibinfo{volume}{60}},
  \bibinfo{pages}{R6275} (\bibinfo{year}{1999}).

\bibitem[{\citenamefont{M{\"u}ller et~al.}(1997)\citenamefont{M{\"u}ller,
  Lippert, K{\"u}hnel, and Behn}}]{MLKB97}
\bibinfo{author}{\bibfnamefont{R.}~\bibnamefont{M{\"u}ller}},
  \bibinfo{author}{\bibfnamefont{K.}~\bibnamefont{Lippert}},
  \bibinfo{author}{\bibfnamefont{A.}~\bibnamefont{K{\"u}hnel}},
  \bibnamefont{and} \bibinfo{author}{\bibfnamefont{U.}~\bibnamefont{Behn}},
  \bibinfo{journal}{Phys. Rev. E} \textbf{\bibinfo{volume}{56}},
  \bibinfo{pages}{2658} (\bibinfo{year}{1997}).

\bibitem[{\citenamefont{K{\"u}rsten and Behn}(2016)}]{KB16}
\bibinfo{author}{\bibfnamefont{R.}~\bibnamefont{K{\"u}rsten}} \bibnamefont{and}
  \bibinfo{author}{\bibfnamefont{U.}~\bibnamefont{Behn}},
  \bibinfo{journal}{Phys. Rev. E} \textbf{\bibinfo{volume}{93}},
  \bibinfo{pages}{033307} (\bibinfo{year}{2016}).

\bibitem[{\citenamefont{Herrmann et~al.}(2012)\citenamefont{Herrmann,
  Niethammer, and Velázquez}}]{HNV12}
\bibinfo{author}{\bibfnamefont{M.}~\bibnamefont{Herrmann}},
  \bibinfo{author}{\bibfnamefont{B.}~\bibnamefont{Niethammer}},
  \bibnamefont{and}
  \bibinfo{author}{\bibfnamefont{J.}~\bibnamefont{Velázquez}},
  \bibinfo{journal}{Multiscale Model. Sim.} \textbf{\bibinfo{volume}{10}},
  \bibinfo{pages}{818} (\bibinfo{year}{2012}).

\bibitem[{\citenamefont{Stuart and Ord}(1994)}]{SO94}
\bibinfo{author}{\bibfnamefont{A.}~\bibnamefont{Stuart}} \bibnamefont{and}
  \bibinfo{author}{\bibfnamefont{J.~K.} \bibnamefont{Ord}},
  \emph{\bibinfo{title}{Kendall's advanced theory of statistics. Distribution
  theory}}, vol.~\bibinfo{volume}{1} (\bibinfo{publisher}{Edward Arnold},
  \bibinfo{address}{London, Melbourne, Auckland}, \bibinfo{year}{1994}),
  \bibinfo{edition}{6th} ed.

\bibitem[{\citenamefont{DasGupta}(2008)}]{dasgupta08}
\bibinfo{author}{\bibfnamefont{A.}~\bibnamefont{DasGupta}},
  \emph{\bibinfo{title}{Asymptotic theory of statistics and probability}}
  (\bibinfo{publisher}{Springer Science \& Business Media},
  \bibinfo{year}{2008}).

\bibitem[{\citenamefont{Dyson}(1943)}]{Dyson43}
\bibinfo{author}{\bibfnamefont{F.~J.} \bibnamefont{Dyson}},
  \bibinfo{journal}{J. R. Stat. Soc.} \textbf{\bibinfo{volume}{106}},
  \bibinfo{pages}{pp. 360} (\bibinfo{year}{1943}).

\bibitem[{\citenamefont{Griffiths et~al.}(1970)\citenamefont{Griffiths, Hurst,
  and Sherman}}]{GHS70}
\bibinfo{author}{\bibfnamefont{R.~B.} \bibnamefont{Griffiths}},
  \bibinfo{author}{\bibfnamefont{C.~A.} \bibnamefont{Hurst}}, \bibnamefont{and}
  \bibinfo{author}{\bibfnamefont{S.}~\bibnamefont{Sherman}},
  \bibinfo{journal}{J. Math. Phys.} \textbf{\bibinfo{volume}{11}},
  \bibinfo{pages}{790} (\bibinfo{year}{1970}).

\bibitem[{\citenamefont{Lebowitz}(1974)}]{Lebowitz74}
\bibinfo{author}{\bibfnamefont{J.}~\bibnamefont{Lebowitz}},
  \bibinfo{journal}{Commun. Math. Phys.} \textbf{\bibinfo{volume}{35}},
  \bibinfo{pages}{87} (\bibinfo{year}{1974}).

\bibitem[{\citenamefont{Ellis and Monroe}(1975)}]{EM75}
\bibinfo{author}{\bibfnamefont{R.~S.} \bibnamefont{Ellis}} \bibnamefont{and}
  \bibinfo{author}{\bibfnamefont{J.~L.} \bibnamefont{Monroe}},
  \bibinfo{journal}{Commun. Math. Phys.} \textbf{\bibinfo{volume}{41}},
  \bibinfo{pages}{33} (\bibinfo{year}{1975}).

\bibitem[{\citenamefont{Ellis et~al.}(1976)\citenamefont{Ellis, Monroe, and
  Newman}}]{EMN76}
\bibinfo{author}{\bibfnamefont{R.~S.} \bibnamefont{Ellis}},
  \bibinfo{author}{\bibfnamefont{J.~L.} \bibnamefont{Monroe}},
  \bibnamefont{and} \bibinfo{author}{\bibfnamefont{C.~M.}
  \bibnamefont{Newman}}, \bibinfo{journal}{Commun. Math. Phys.}
  \textbf{\bibinfo{volume}{46}}, \bibinfo{pages}{167} (\bibinfo{year}{1976}).

\bibitem[{\citenamefont{Ellis and Newman}(1978)}]{EN78}
\bibinfo{author}{\bibfnamefont{R.~S.} \bibnamefont{Ellis}} \bibnamefont{and}
  \bibinfo{author}{\bibfnamefont{C.~M.} \bibnamefont{Newman}},
  \bibinfo{journal}{T. Am. Math. Soc.} \textbf{\bibinfo{volume}{237}},
  \bibinfo{pages}{83} (\bibinfo{year}{1978}).

\end{thebibliography}

\end{document}